\documentclass[11pt]{article}
\usepackage[utf8]{inputenc}
\usepackage[T1]{fontenc}
\usepackage[english]{babel}
\usepackage{amssymb,amsmath,amsthm}
\usepackage{graphicx}
\usepackage{algorithm}
\usepackage{algorithmicx}
\usepackage[noend]{algpseudocode}
\usepackage{authblk}
\usepackage{thmtools}
\usepackage{etoolbox}
\usepackage{thm-restate}
\usepackage{flushend}
\usepackage{fullpage}
\usepackage{subcaption}
\usepackage{microtype}
\usepackage{textcomp}

\usepackage[usenames,dvipsnames]{color}
 
\newtheorem{theorem}{Theorem}[section]
\newtheorem{lemma}[theorem]{Lemma}
\newtheorem{corollary}[theorem]{Corollary}
\newtheorem{Definition}[theorem]{Definition}
\newtheorem{observation}[theorem]{Observation}

\newcommand{\important}[1]{{\textcolor{black}{#1}}}

\newcommand{\mytitle}[1]{\begin{center}\huge\textbf{#1}\end{center}}
\newcommand{\myauthors}[1]{\begin{center}\large#1\end{center}}
\newcommand{\mydate}[1]{\begin{center}\large#1\end{center}}

\newcommand{\myinst}[2]{\makebox[0pt][l]{\hspace{1pt}$^\text{\itshape #1}$}#2\hspace{5pt}}

\newcommand{\myabstract}[1]{\noindent\textbf{Abstract.}#1 \par}

\newcommand{\myendtitlepage}{\thispagestyle{empty}\setcounter{page}{0}\newpage}

\renewcommand{\S}{Section} % Replace the paragraph sign by "Section".
\newcommand{\old}[1]{{}}

\newcommand{\revised}[1]{{\color{black} #1}}

\begin{document}

%%%%%%%%%%%%%%%%%%%%%%%%%%%%%%%%
\mytitle{Conflict-Free Coloring of Graphs}
\medskip
\myauthors{%
    Zachary~Abel\myinst{a},
    Victor~Alvarez\myinst{b}, 
    Erik~D.~Demaine\myinst{c},
    S{\'a}ndor~P.~Fekete\myinst{b},
    Aman Gour\myinst{d}, \\[2pt]
    Adam Hesterberg\myinst{a}, 
    Phillip Keldenich\myinst{b}, and
    Christian Scheffer\myinst{b}{}
}
\smallskip
\mydate{}
\bigskip
%\old{
{\small
\begin{itemize}
\item[\small $^\text{\itshape a}$]{%
        Mathematics Department, MIT, Cambridge, Massachusetts, United States, 
    }{%
        \par zabel@mit.edu, achesterberg@gmail.com
    }
\item[\small $^\text{\itshape b}$]{%
        Department of Computer Science,
        Braunschweig University of Technology%
    }{%
        \par \{s.fekete,v.alvarez,p.keldenich,c.scheffer\}@tu-bs.de%
    }
\item[\small $^\text{\itshape c}$]{%
        CSAIL, MIT, Cambridge, Massachusetts, United States, 
    }{%
        \par edemaine@mit.edu
    }
\item[\small $^\text{\itshape d}$]{%
        Department of Computer Science and Engineering,
        IIT Bombay%
    }{%
        \par amangour30@gmail.com
    }
\end{itemize}
}
%}%old
\bigskip
\bigskip
\myabstract{
A \emph{conflict-free $k$-coloring} of a graph assigns one of $k$ different
colors to some of the vertices such that, for every vertex~$v$, there is a color
that is assigned to exactly one vertex among $v$ and $v$'s neighbors.
Such colorings have applications in wireless networking, robotics, and geometry,
and are well-studied in graph theory.
Here we study the natural problem of the \emph{conflict-free chromatic number} $\chi_{CF}(G)$
(the smallest $k$ for which conflict-free $k$-colorings exist). We provide results
both for {\em closed} neighborhoods $N[v]$, for which a vertex $v$ is a member of its neighborhood,
and for {\em open} neighborhoods $N(v)$, for which vertex $v$ is not a member of its neighborhood.

%with a focus on planar graphs.

For closed neighborhoods, we prove the conflict-free variant of the famous Hadwiger Conjecture:
If an arbitrary	graph $G$ does not contain $K_{k+1}$ as a minor, then $\chi_{CF}(G) \leq k$.
For planar graphs, we obtain a tight worst-case bound:
three colors are sometimes necessary and always sufficient.
In addition, we give a complete characterization of the
algorithmic/computational complexity of conflict-free coloring.
It is NP-complete to decide whether a planar graph has a conflict-free coloring
with {\em one} color, while for outerplanar graphs, this can be decided in polynomial time.
Furthermore, it is NP-complete to decide whether a planar graph has a
conflict-free coloring with {\em two} colors, while for outerplanar graphs, two colors always suffice.
For the \important{bicriteria} problem of minimizing the number of colored vertices subject to
a given bound $k$ on the number of colors, we give a full algorithmic
characterization in terms of complexity and approximation for outerplanar and
planar~graphs. 

For open neighborhoods, we show that every planar {\em bipartite} graph has a
conflict-free coloring with at most four colors; on the other hand, we prove that for $k\in \{1,2,3\}$, it is NP-complete
to decide whether a planar bipartite graph has a conflict-free $k$-coloring.
Moreover, we establish that any {\em general} planar graph has
a conflict-free coloring with at most eight colors. 
}

\noindent\makebox[\linewidth]{\rule{\textwidth}{0.4pt}} 
{An extended abstract containing major parts of this paper
was entitled ``Three colors suffice: Conflict-free coloring of planar graphs'' and appeared in the 
Proceedings of the Twenty-Eighth Annual ACM-SIAM Symposium on Discrete Algorithms (SODA 2017)~\cite{aad+-tcscf-17}.}

\myendtitlepage

%%%%%%%%%%%%%%%%%%%%%%%%%%%%%%%%

%\addtocounter{page}{-1}
% !TEX root = ./main.tex
\section{Introduction}
\label{sec:intro}
Coloring the vertices of a graph is one of the fundamental problems in graph theory, both scientifically and historically.
Proving that four colors always suffice to color a planar graph \cite{AH1,AH2,RSST} was a tantalizing open problem for
more than 100 years; the quest for solving this challenge contributed to the development of graph theory, but also to 
computers in theorem proving \cite{wilson}. A generalization that is still unsolved is the \emph{Hadwiger Conjecture}~\cite{hadwiger}:
A graph is $k$-colorable if it has no $K_{k+1}$~minor.

Over the years, there have been many variations on coloring, often motivated by particular applications.
One such context is wireless communication, where ``colors'' correspond to different frequencies.
This also plays a role in robot navigation, where different beacons are used for providing direction.
To this end, it is vital that in any given location, a robot is adjacent to a beacon with a frequency
that is unique among the ones that can be received. \important{This notion has been introduced as \emph{conflict-free coloring},
formalized as follows. For any vertex $v\in V$ of a simple graph \important{$G=(V,E)$}, 
the \emph{closed neighborhood} $N[v]$ consists of all vertices adjacent to $v$ and $v$ itself.
A \emph{conflict-free $k$-coloring} of $G$ assigns one of $k$ different colors to a (possibly proper) subset $S \subseteq V$ of vertices, such that for every vertex $v\in V$, there is a vertex $y \in N[v]$, called the \emph{conflict-free neighbor} of $v$, such that %there is no $y' \neq y \in N[v]$ with $c(y) = c(y')$, i.e. 
the color of $y$ is unique in the closed neighborhood of $v$.}
The \emph{conflict-free chromatic number} $\chi_{CF}(G)$ of $G$ is the
smallest $k$ for which a conflict-free coloring exists.
\important{Observe that $\chi_{CF}(G)$ is bounded from above by the proper chromatic number $\chi(G)$ because in a proper coloring, every vertex is its own conflict-free~neighbor.}

Similar questions can be considered for {\em open neighborhoods} $N(v)=N[v]\setminus \{v\}$.

Conflict-free coloring has received an increasing amount of attention.
Because of the relationship to classic coloring, it is natural to investigate the conflict-free coloring of planar graphs.
In addition, previous work has considered either general graphs and hypergraphs (e.g., see \cite{pt-cfcgh-09})
or geometric scenarios (e.g., see \cite{hoffmann_et_al:LIPIcs:2015:5097}); we give a more detailed overview further down. 
This adds to the relevance of conflict-free coloring of planar graphs,
which constitute the intersection of general graphs and geometry.
In addition, the subclass of outerplanar graphs is of interest, as it corresponds to
subdividing simple polygons by chords.

There is a spectrum of different scientific challenges when studying conflict-free
coloring. What are worst-case bounds on the necessary number of colors? When is it NP-hard to determine the 
existence of a conflict-free $k$-coloring, when polynomially solvable?
What can be said about approximation? Are there sufficient conditions for more general graphs?
And what can be said about the \important{bicriteria} problem, in which also the number of colored vertices
is considered? We provide extensive answers for all of these aspects, basically providing a complete
characterization for planar and outerplanar graphs.

\subsection{Our Contribution} We present the following results; items 1-7 are for closed neighborhoods, while items 8-11 are for open neighborhoods.

\begin{enumerate}
\item For general graphs, we provide the conflict-free variant of the Hadwiger Conjecture:
If $G$ does not contain $K_{k+1}$ as a minor, then $\chi_{CF}(G) \leq k$.

\item It is NP-complete to decide whether a planar graph has a conflict-free coloring
with \emph{one} color. For outerplanar graphs, this question can be decided in polynomial time.

\item It is NP-complete to decide whether a planar graph has a
conflict-free coloring with \emph{two} colors. For outerplanar graphs, two colors always suffice.

\item Three colors are sometimes necessary and always sufficient for conflict-free coloring of a planar graph.

\item For the \important{bicriteria} problem of minimizing the number of colored vertices subject to
a given bound $\chi_{CF}(G)\leq k$ with $k\in\{1,2\}$, we prove that the problem is NP-hard for planar 
and polynomially solvable in outerplanar graphs.

\item For planar graphs and $k=3$ colors, minimizing the number of colored vertices does not have a constant-factor approximation, unless P = NP.

\item For planar graphs and $k \geq 4$ colors, it is NP-complete to minimize the number of colored vertices. The problem is fixed-parameter tractable (FPT) and allows a PTAS.

\item Four colors are sometimes necessary and always sufficient for conflict-free coloring with open neighborhoods of planar bipartite graphs.
\item It is NP-complete to decide whether a planar bipartite graph has a conflict-free coloring with open neighborhoods with $k$ colors for $k \in \{1,2,3\}$.
\item Eight colors always suffice for conflict-free coloring with open neighborhoods of planar graphs.
\end{enumerate}

\subsection{Related Work}
In a geometric context, the study of conflict-free coloring was started by
Even, Lotker, Ron, and Smorodinsky~\cite{elrs-cfcsg-03} and
Smorodinsky~\cite{s-cpcg-03}, who motivate the problem by frequency
assignment in cellular networks: There, a set of $n$ base stations is given, each covering some geometric region in the plane.
The base stations service mobile clients that can be at any point in the total covered area.
To avoid interference, there must be at least one base station in range using a unique frequency for every point in the entire covered area.
The task is to assign a frequency to each base station minimizing the number of frequencies.
On an abstract level, this induces a coloring problem on a hypergraph where the base stations correspond to the vertices and there is an hyperedge between some vertices if the range of the corresponding base stations has a non-empty common intersection.

If the hypergraph is induced by disks, Even et al.~\cite{elrs-cfcsg-03} prove that $\mathcal{O}(\log n)$ colors are always sufficient.
Alon and Smorodinsky~\cite{as-cfcsd-06} extend this by showing that each family of disks, where each disk intersects at most $k$ others, can be colored using $\mathcal{O}(\log^3 k)$ colors.
Furthermore, for unit disks, Lev-Tov and Peleg~\cite{lp-cfcud-09} present an $\mathcal{O}(1)$-approximation algorithm for the number of colors.
Horev et al.~\cite{hks-cfcms-10} extend this by showing that any set of $n$ disks can be colored with $\mathcal{O}(k\log n)$ colors, even if every point must see $k$ distinct unique colors.
Abam et al.~\cite{abp-ftcfc-08} discuss the problem in the context of cellular networks where the network has to be reliable even if some number of base stations fault, giving worst-case bounds for the number of colors required.

For the dual problem of coloring a set of points such that each region from some family of regions contains at least one uniquely colored point, Har-Peled and Smorodinsky~\cite{hs-cfcpsrp-05} prove that with respect to every family of pseudo-disks, every set of points can be colored using $\mathcal{O}(\log n)$ colors.
For rectangle ranges, Elbassioni and Mustafa~\cite{em-cfcrr-06} show that it is possible to add a sublinear number of points such that a conflict-free coloring with $\mathcal{O}(n^{3/8 \cdot (1+\varepsilon)})$ colors becomes possible.
Ajwani et al.~\cite{aegr-cfcrro-07} complement this by showing that coloring a set of points with respect to rectangle ranges is always possible using $\mathcal{O}(n^{0.382})$ colors.
For coloring points on a line with respect to intervals, Cheilaris et al.~\cite{cgrs-scfci-14} present a 2-approximation algorithm, and a $\left(5-\frac{2}{k}\right)$-approximation algorithm when every interval must see $k$ uniquely colored vertices.
%(Even et al. (TODO cite) discuss the problem
%of minimum conflict-free coloring (min-CF-coloring) in general congruent
%centrally symmetric convex regions in the plane. )
Hoffman et al.~\cite{hoffmann_et_al:LIPIcs:2015:5097} give tight bounds for the
conflict-free chromatic art gallery problem under rectangular visibility in
orthogonal polygons: $\Theta(\log\log n)$ are sometimes necessary and always sufficient.
Chen et al.~\cite{cfk+-ocfci-07} consider the online version of the conflict-free coloring
of a set of points on the line, where each newly inserted point must be
assigned a color upon insertion, and at all times the coloring has to be
conflict-free.
Also in the online scenario, Bar-Nov et al.~\cite{bco+-ocfch-10} consider a certain class of $k$-degenerate hypergraphs which sometimes arise as intersection graphs of geometric objects, presenting an online algorithm using $\mathcal{O}(k\log n)$ colors.

On the combinatorial side, some authors consider the variant in which all vertices need to be colored; note that this does not change asymptotic results for general
graphs and hypergraphs: it suffices to introduce one additional color for vertices that are left uncolored in our constructions.
Regarding general hypergraphs, Ashok et al.~\cite{adk-efmcfch-15} prove that maximizing the number of conflict-freely colored edges in a hypergraph is FPT when parameterized by the number of conflict-free edges in the solution.
Cheilaris et al.~\cite{css-piclcfcgh-11} consider the case of hypergraphs induced by a set of planar Jordan regions and prove an asymptotically tight upper bound of $\mathcal{O}(\log n)$ for the conflict-free list chromatic number of such hypergraphs.
They also consider hypergraphs induced by the simple paths of a planar graph and prove an upper bound of $\mathcal{O}(\sqrt{n})$ for the conflict-free list chromatic number.
For hypergraphs induced by the paths of a simple graph $G$, Cheilaris and T\'oth~\cite{ct-gumcfc-11} prove that it is coNP-complete to decide whether a given coloring is conflict-free if the input is $G$.
Regarding the case in which the hypergraph is induced by the neighborhoods of a simple graph $G$, which resembles our scenario, Pach and T\'ardos~\cite{pt-cfcgh-09} prove that the conflict-free chromatic number of an $n$-vertex graph is in $\mathcal{O}(\log^2 n)$.
Glebov et al.~\cite{gst-cfcg-14} extend this from an extremal and probabilistic point of view by proving that almost all $G(n,p)$-graphs have conflict-free chromatic number $\mathcal{O}(\log n)$ for $p \in \omega(1/n)$, and by giving a randomized construction for graphs having conflict-free chromatic number $\Theta(\log^2 n)$.
In more recent work, Gargano and Rescigno~\cite{gr-ccfcg-15} show that finding the conflict-free chromatic number for general graphs is NP-complete, and prove that the problem is FPT w.r.t.~vertex cover or neighborhood~diversity~number.
%Our work differs from most of the above works by relaxing the requirement of having to color every vertex of the underlying graph.

%All loopless outerplanar graphs can be colored using
%only three colors;\cite{three-color} this fact features prominently in the
%simplified proof of Chvátal's art gallery theorem by Fisk (1978). A 3-coloring
%may be found in linear time by a greedy coloring algorithm that removes any
%vertex of degree at most two, colors the remaining graph recursively, and then
%adds back the removed vertex with a color different from the colors of its two
%neighbors. 

%\paragraph{Outline of the paper.} In Section \ref{sec:conflict-free-coloring-of-general-graphs}, we discuss the problem of conflict-free coloring ... TODO

% !TEX root = ./main.tex
\section{Preliminaries}
% Let $G = (V, E)$ be a simple graph. For every $x \in V$ we denote by $N(x) = {y \in V : (x,y) \in E}$ its neighborhood and by $N[x] = N(x)\cup{x}$ its closed neighborhood. A (not necessarily proper) vertex coloring $\chi$ of G is called conflict-free, if for each vertex x $\in$ V, there exists a vertex y in $N[x]$ whose color is different from the color of each other vertex in $N[x]$. We then say that y has unique color in $N[x]$. The conflict-free chromatic number $\chi_{CF}(G)$ is the smallest r, such that there exists a conflict-free r-coloring of G. Conflict-free coloring can be interpreted as a relaxation of the usual proper coloring concept where each vertex x is required to have a unique color in its own closed neighborhood $N[x]$.

{For every vertex $v \in V$, the \emph{open neighborhood} of $v$ in $G$ is denoted by $N_G(v) := \{ w \in V(G) \,|\, vw \in E(G) \}$, and the \emph{closed neighborhood} is denoted by $N_G[v] := N_G(v) \cup \{v\}$.
We sometimes write $N(v)$ instead of $N_G(v)$ when $G$ is clear from the context.}

A partial $k$-coloring of $G$ is an assignment $\chi: V' \to \{1,\ldots,k\}$ of colors to a subset $V' \subseteq V(G)$ of the vertices.
$\chi$ is called \emph{closed-neighborhood conflict-free $k$-coloring} of $G$ iff, for each vertex $v \in V$, there is a vertex $w \in N_G[v] \cap V'$ such that $\chi(w)$ is unique in $N_G[v]$, i.e., for all other $w' \in N_G[v] \cap V'$, $\chi(w') \neq \chi(w)$.
We call $w$ the conflict-free neighbor of $v$.
{Analogously, $\chi$ is called \emph{open-neighborhood conflict-free $k$-coloring} of $G$ iff, for each vertex $v \in V$, there is a conflict-free neighbor $w \in N_G(v)$.}

In order to avoid confusion with \emph{proper $k$-colorings}, i.e., colorings that color all vertices such that no adjacent vertices receive the same color, we use the term \emph{proper coloring} when referring to this kind of coloring.
{The minimum number of colors needed for a proper coloring of $G$, also known as the chromatic number of $G$, is denoted by $\chi_P(G)$, whereas the minimum number of colors required for a closed-neighborhood conflict-free coloring of $G$ ($G$'s \emph{closed-neighborhood conflict-free chromatic number}) is written as $\chi_{CF}(G)$.
The \emph{open-neighborhood conflict-free chromatic number} of $G$ is $\chi_O(G)$.
To improve readability we sometimes omit the type of neighborhood if it is clear from the context.}

Note that, because every vertex satisfies $v \in N[v]$, every proper coloring
of $G$ is also a {closed-neighborhood} conflict-free coloring of $G$, and thus
$\chi_{CF}(G) \leq \chi_{P}(G)$.  {The same does not hold for open
neighborhoods. There is no constant factor $c_1 > 0$ such that either $c_1
\cdot \chi_O(G) \leq \chi_P(G)$ or $c_1 \cdot \chi_P(G) \leq \chi_O(G)$ holds
for all graphs $G$.}

{For closed neighborhoods,} we define the \emph{conflict-free domination number} $\gamma_{CF}^k(G)$ of $G$ to be the minimum number of vertices that have to be colored in a conflict-free $k$-coloring of $G$.
We set $\gamma_{CF}^k(G) = \infty$ if $G$ is not conflict-free $k$-colorable.
Because the set of colored vertices is a dominating set, the conflict-free domination number satisfies $\gamma_{CF}^k(G) \geq \gamma(G)$ for all $k$, where $\gamma(G)$, the domination number of $G$, is the size of a minimum dominating set of $G$.
Moreover, for any graph, there is a $k \leq \gamma(G)$ such that $\gamma_{CF}^k(G) = \gamma(G)$.

\begin{sloppypar}
We denote the complete graph on $n$ vertices by $K_n\important{ :=
(\{1,\ldots,n\},\{\{u,v\}\,|\, u,v \in \{1,\ldots,n\}}$, $u \neq v\})$, and the
complete bipartite graph on $n$ and $m$ vertices as $K_{n,m}$. We define the
graph $K_{n}^{-3} := (V(K_{n}), E(K_{n})\setminus E(K_{3}))$, which is obtained
by removing any three edges forming a single triangle from a $K_n$.
\end{sloppypar}

We also provide a number of results for outerplanar graphs. 
An outerplanar graph is a graph that has a planar embedding for which all vertices belong to the outer face of the embedding.
An outerplanar graph is called \emph{maximal} iff no edges can be added to the graph without losing outerplanarity \cite{outerplanar}.
Maximal outerplanar graphs can also be characterized as the graphs having an embedding corresponding to a polygon~triangulation, which illustrates their particular relevance in a geometric context.
In addition, maximal outerplanar graphs exhibit a number of interesting graph-theoretic properties.
Every maximal outerplanar graph is chordal, a 2-tree and a series-parallel graph. 
Also, every maximal outerplanar graph is the visibility graph of a simple polygon. 
%An outerplanar graph is a graph that has a planar embedding for which all vertices belong to the outer face of the embedding.
%An outerplanar graph is called \emph{maximal} iff no edges can be added to the graph without losing outerplanarity \cite{outerplanar}.
%Maximal outerplanar graphs exhibit interesting properties.
%Every maximal outerplanar graph is chordal, a 2-tree and a series-parallel graph.
%Every maximal outerplanar graph is the visibility graph of a simple polygon.
%Maximal outerplanar graphs can also be characterized as the graphs having an embedding corresponding to a polygon triangulation.

%The famous \emph{four color theorem} proves that any planar graph can be properly colored using four colors.
%A trivial corollary to the above result is the fact that any planar graph can be colored using just 4 colors in a conflict-free manner.

For some of our NP-hardness proofs, we use a variant of the planar 3-SAT problem, called {\sc Positive Planar 1-in-3-SAT}.
This problem was introduced and shown to be NP-complete by Mulzer and Rote \cite{mr-mwtnph-08}, and consists of deciding whether a given positive planar 3-CNF formula allows a truth assignment such that in each clause, exactly one literal is true.

\begin{Definition}[Positive planar formulas]\ \\
A formula $\phi$ in 3-CNF is called \emph{positive planar} iff it is both \emph{positive} and \emph{backbone planar}.
A formula $\phi$ is called \emph{positive} iff it does not contain any negation, i.e. iff all occurring literals are positive.
A formula $\phi$, with clause set $C = \{c_1,\ldots,c_l\}$ and variable set $X = \{x_1,\ldots,x_n\}$, is called \emph{backbone planar} iff its associated graph $G(\phi) := (X \cup C, E(\phi))$ is planar, where $E(\phi)$ is defined as follows:
\begin{itemize}
	\item $x_ic_j \in E(\phi)$ for a clause $c_j \in C$ and a variable $x_i \in X$ iff $x_i$ occurs in $c_j$,
	\item $x_ix_{i+1} \in E(\phi)$ for all $1 \leq i < n$.
\end{itemize}
The path formed by the latter edges is also called the \emph{backbone} of the formula graph $G(\phi)$.
\end{Definition}

% !TEX root = ./main.tex
\section{Closed Neighborhoods: Conflict-Free Coloring of General Graphs}
\label{sec:conflict-free-coloring-of-general-graphs}
In this section we consider the {\sc Conflict-Free $k$-Coloring} problem on general simple graphs {with respect to closed neighborhoods}.
In \S~\ref{sec:general-graphs-complexity}, we prove that this problem is NP-complete for any $k \geq 1$.
In \S~\ref{sec:general-graphs-sufficiency}, we provide a sufficient criterion that guarantees conflict-free $k$-colorability.
In \S~\ref{sec:general-graphs-complexity-domination-number}, we consider the conflict-free domination number and prove that, for any $k \geq 3$, there is no constant-factor approximation algorithm for $\gamma^k_{CF}$.

\subsection{Complexity}
\label{sec:general-graphs-complexity}
\begin{restatable}{theorem}{generalnpc}
\label{thm:conflict-free-coloring-npc}
{\sc Conflict-Free $k$-Coloring} is NP-complete for any fixed $k \geq 1$.
\end{restatable}

\noindent Membership in NP is clear.
For $k \geq 3$, we prove NP-hardness using a reduction from proper {\sc $k$-Coloring}.
For $k \in \{1,2\}$, refer to \S~\ref{sec:conflict-free-coloring-planar}, where we prove {\sc Conflict-Free $k$-Coloring} of planar graphs to be NP-complete for $k \in \{1,2\}$.

Central to the proof is the following lemma that enables us to enforce certain vertices to be colored, and both ends of an edge to be colored using distinct colors.

\begin{lemma}
\label{lem:inductive-coloring-gadgets}
Let $G$ be any graph, $u,v \in V(G)$ and $vu = e \in E(G)$.
If $N(v)$ contains two disjoint and independent copies of a graph $H$ with $\chi_{CF}(H) = k$, not adjacent to any other vertex $w \in G$, every conflict-free $k$-coloring of $G$ colors $v$.
If the same holds for $u$ and in addition, $N_G(u) \cap N_G(v)$ contains two disjoint and independent copies of a graph $J$ with $\chi_{CF}(J) = k-1$, not adjacent to any other vertex $w \in G$, every conflict-free $k$-coloring of $G$ colors $u$ and $v$ with different colors.
\begin{proof}
Assume towards a contradiction that there was a conflict-free $k$-coloring $\chi$ that avoids coloring $v$.
Then, due to the copies of $H$ being independent, disjoint and not connected to any other vertex, the restriction of $\chi$ to the vertices of each of the two copies must induce a conflict-free coloring on $H$.
As $\chi_{CF}(H) = k$, this implies that $\chi$ uses $k$ colors on each copy.
Therefore, in the open neighborhood of $v$, there are at least two vertices colored with each color.
This leads to a contradiction, because $v$ cannot have a conflict-free~neighbor.

For the second proposition, suppose there was a conflict-free coloring assigning the same color to $u$ and $v$.
Without loss of generality, let this color be 1.
As every vertex of the two copies of $J$ now sees two occurrences of color 1, color 1 can not be the color of the unique neighbor of any vertex of $J$, and any occurrence of color 1 on the vertices of $J$ can be removed.
Therefore, we can assume each of the two copies of $J$ to be colored in a conflict-free manner using the colors $\{2,\ldots,k\}$.
Observe that, due to $\chi_{CF}(J) = k-1$, each of these colors must be used at least once in each copy.
This implies that both $u$ and $v$ see each color at least twice: The two copies of $J$ enforce two occurrences of the colors $\{2,\ldots,k\}$, and color 1 is assigned to both $u$ and $v$, which are connected by an edge.
This is a contradiction, and therefore, both $u$ and $v$ must be colored with distinct colors.
\end{proof}
\end{lemma}
\noindent Next, we give an inductive construction of graphs, $G_k$, with $\chi_{CF}(G_k) = k$. The proof of NP-hardness relies on this hierarchy.
\begin{enumerate}
\item The first graph $G_1$ of the hierarchy consists of a single isolated vertex.
$G_2$ is a $K_{1,3}$ with one edge subdivided by another vertex, or, equivalently, a path of length 3 with a leaf vertex attached to one of the inner vertices.
\item Given $G_k$ and $G_{k-1}$, $G_{k+1}$ is constructed as follows for $k \geq 2$:
\begin{itemize}
\item Take a complete graph $G = K_{k+1}$ on \mbox{$k+1$}~vertices.
\item To each vertex $v \in V(K_{k+1})$, attach two disjoint and independent copies of $G_k$, adding an edge from $v$ to every vertex of both copies of $G_k$.
\item For each edge $e = vw \in E(K_{k+1})$, add two disjoint and independent copies of $G_{k-1}$, adding an edge from $v$ and $w$ to every vertex of both copies.
\end{itemize}
\end{enumerate}
The number of vertices of the graphs $G_k$ obtained by the above construction satisfies the recursive formula $$|G_1| = 1, |G_2| = 5, |G_{k+1}| = (k+1) \cdot (2|G_k| + k|G_{k-1}| + 1),$$ which is in $\Omega\left(2^k\right)$ and $\mathcal{O}\left(2^{k\log k}\right)$. Figure~\ref{fig:cf-graph-g3} depicts the graph $G_3$, which in addition to being planar is a series-parallel graph.

\begin{figure}[h]
\centering
\includegraphics[width=.6\linewidth,keepaspectratio]{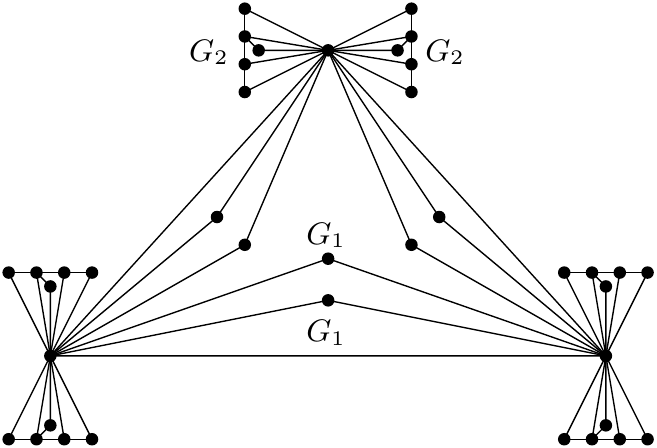}
\caption{The graph $G_3$.}
\label{fig:cf-graph-g3}
\end{figure}

\begin{lemma}
For $G_{k}$ constructed in this manner, $\chi_{CF}(G_{k}) = k$.
\begin{proof}
The proof uses induction over $k$.
Application of Lemma~\ref{lem:inductive-coloring-gadgets} implies that all vertices of the $K_{k+1}$ underlying $G_{k+1}$ have to be colored using different colors.
Therefore, $\chi_{CF}(G_{k+1}) \geq k+1$.
By coloring all $k+1$ vertices of the underlying $K_{k+1}$ with a different color, we obtain a conflict-free $(k+1)$-coloring of $G_{k+1}$, implying $\chi_{CF}(G_{k+1}) \leq k+1$.
\end{proof}
\end{lemma}
\begin{lemma}
For $k \geq 2$, {\sc $k$-Coloring $\preccurlyeq$ Conflict-Free $k$-Coloring}.
Therefore, for $k \geq 3$, {\sc Conflict-Free $k$-Coloring} is NP-complete.
\begin{proof}
Given a graph $G$ for which to decide proper $k$-colorability for a fixed $k$. We construct a graph $G'$ that is conflict-free $k$-colorable iff $G$ is $k$-colorable.
$G'$ is constructed from $G$ by attaching two copies of $G_k$ to each vertex $v \in V(G)$, by adding an edge from $v$ to each vertex of the copies of $G_k$.
For each edge $uv \in E(G)$, we attach two copies of $G_{k-1}$ to both endpoints of $uv$ by adding an edge from $u$ and $v$ to all vertices of both copies.
As $k$ is fixed, $|G_k|$ and $|G_{k-1}|$ are constant, implying that $G'$ can be constructed in polynomial time.

A proper $k$-coloring of $G$ induces a conflict-free $k$-coloring of $G'$ by leaving all other vertices uncolored.
On the other hand, by Lemma~\ref{lem:inductive-coloring-gadgets}, a conflict-free $k$-coloring $\chi$ of $G'$ colors all vertices $v \in V(G)$ and for every edge, the colors of both endpoints are distinct.
Therefore, the restriction of $\chi$ to $V(G)$ is a proper $k$-coloring of $G$.
\end{proof}
\end{lemma}

\subsection{A Sufficient Criterion for $k$-Colorability}\label{sec:general-graphs-sufficiency}
In this section we present a sufficient criterion for conflict-free $k$-colorability together with an efficient heuristic that can be used to color graphs satisfying this criterion with $k$ colors in a conflict-free manner.
This heuristic is called \emph{iterated elimination of distance-3-sets} and is detailed in Algorithm \ref{alg:iterated-elimination-distance-3-sets}.
The main idea of this heuristic is to iteratively compute maximal sets of vertices at pairwise (link) distance at least 3, coloring all vertices in one of these sets using one color, and then removing these vertices and their neighbors until all that remains is a collection of disconnected paths, which can then be colored using one color.
\begin{algorithm*}	
	\begin{algorithmic}[1]
		\State{$i \leftarrow 1, \chi \leftarrow \emptyset$}
		\State Remove all isolated paths from $G$
		\While{$G$ is not empty}
			\State $D \leftarrow \emptyset$
			\State For each component of $G$, select some vertex $v$ and add it to $D$
			\While{there is a vertex $w$ at distance $\geq 3$ from all vertices in $D$}
				\State{Choose $w$ at distance exactly 3 from some vertex in $D$}
				\State{$D \leftarrow D \cup \{w\}$}
			\EndWhile
			\State $\forall u \in D: \chi(u) \leftarrow i$
			\State $i \leftarrow i+1$
			\State Remove $N[D]$ from $G$
			\State Remove all isolated paths from $G$
		\EndWhile
		\State Color all removed isolated paths using color $i$
	\end{algorithmic}
	\caption{Iterated elimination of distance-3-sets}
	\label{alg:iterated-elimination-distance-3-sets}
\end{algorithm*}

\begin{theorem}
Let $G$ be a graph and $k \geq 1$.
If $G$ has neither $K_{k+2}$ nor $K_{k+3}^{-3}$ as a minor, $G$ admits a conflict-free $k$-coloring that can be found in polynomial time using iterated elimination of distance-3 sets.
\label{thm:sufficient-criterion-excluded-minors}
\end{theorem}
\begin{proof}
For $k=1$, a graph $G$ with neither a $K_3$ nor a $K_4^{-3} = K_{1,3}$ minor consists of a collection of isolated paths.
A path on $3n$ vertices can be colored with one color by coloring the middle vertex of every three vertices.
This does not color the vertices at either end, so up to two vertices can be removed from the path to get colorings for paths on $3n-1$ and $3n-2$ vertices.

For $k \geq 2$, we use induction as follows:
First, we color an inclusion-wise maximal subset $D \subseteq V$ of vertices at pairwise distance at least 3 to each other using color 1.
\revised{This set $D$ is chosen such that each vertex $v \in D$ is at distance exactly 3 from some $v' \in D$.
Coloring $D$} provides a conflict-free neighbor of color 1 to every vertex in $N[D]$.
Therefore, the vertices in $N[D]$ are covered and can be removed from the graph.
The remaining graph consists of vertices at distance 2 to some vertex in $D$; we call these vertices \emph{unseen} in the remainder of the proof.
We show that the remaining graph has no $K_{k+1}$ and no $K_{k+2}^{-3}$ as a minor.
By induction, iterated elimination of distance 3 sets requires $k-1$ colors to color the remaining graph, and thus $k$ colors suffice for $G$.

If the graph is disconnected, iterated elimination of distance 3 sets works on all components separately, so we can assume $G$ to be connected.
We claim that there is no set $U$ of unseen vertices that is a cutset of $G$.
Suppose there were such a cutset $U$ and let $H$ be any component of $G \setminus U$ not containing $v$, the first selected vertex during the construction of $D$.
At least one vertex of $H$ is colored: every vertex in $U$ is at distance at least two from every colored vertex not in $H$, therefore, every vertex in $H$ is at distance at least three from every colored vertex not in $H$.
Consider the iteration where the first vertex $w$ of $H$ is added to the set of colored vertices $D$.
At this point, $w$ is at distance exactly 3 from some colored vertex not in $H$.
However, this implies $w$ is adjacent to some vertex from $U$, contradicting the fact that all vertices in $U$ are unseen.

Now, suppose for the sake of contradiction that the set $W$ of unseen vertices contains a $K_{k+1}$ or $K_{k+2}^{-3}$ minor.
$W$ is not the whole graph, because at least one vertex is colored, so there must be a vertex $v$ not in the $K_{k+1}$ or $K_{k+2}^{-3}$ minor.
For every vertex $w \in W$, there is a path from $v$ to $w$ that intersects $W$ only at $w$.
Otherwise, $W \setminus \{w\}$ would be a cutset separating $v$ from $w$.
So, if the graph induced by $W$ had a $K_{k+1}$ or $K_{k+2}^{-3}$ minor, we could contract $G \setminus W$ to a single vertex, which would be adjacent to all vertices in $W$, yielding a $K_{k+2}$ or $K_{k+3}^{-3}$ minor of $G$, which does not exist.
\end{proof}
\noindent Observe that $G_{k+1}$ contains a $K_{k+3}^{-3}$ as a minor, but not a $K_{k+2}$, proving that just excluding $K_{k+2}$ as a minor does not suffice to guarantee $k$-colorability.
Moreover, note that $K_{k+1}$ is a minor of $K_{k+2}$ and $K_{k+3}^{-3}$.

This yields the following corollary, which is the conflict-free variant of the Hadwiger Conjecture.

\begin{corollary}
	All graphs that do not have $K_{k+1}$ as a minor are conflict-free $k$-colorable.
\end{corollary}

%\important{Thus, we prove the conflict-free variant of the Hadwiger Conjecture, which states that all graphs with no $K_{k+1}$-minor are conflict-free $k$-colorable.}

\subsection{Conflict-Free Domination Number}\label{sec:general-graphs-complexity-domination-number}
In this section we consider the problem of minimizing the number of colored vertices in a conflict-free $k$-coloring for a fixed $k$, which is equivalent to computing $\gamma^{k}_{CF}$. We call the corresponding decision problem {\sc{$k$-Conflict-Free Dominating Set}}. We show that approximating the conflict-free domination number in general graphs is hard for any fixed $k$. In \S~\ref{sec:bicriteria-conflict-free-coloring-planar} we discuss the {\sc{$k$-Conflict-Free Dominating Set}} problem for planar graphs.
\begin{theorem}
\label{thm:inapproximability-of-gamma-cf-k}
Unless $\mbox{P} = \mbox{NP}$, for any $k \geq 3$, there is no polynomial-time approximation algorithm for $\gamma^k_{CF}(G)$ with constant approximation factor.
\begin{proof}
We use a reduction from proper {\sc $k$-Coloring} for the proof.
Assume towards a contradiction that there was a polynomial-time approximation algorithm for $\gamma^k_{CF}(G)$ with approximation factor $c \geq 1$.
Let $G$ be a graph on $n$ vertices for which we want to decide $k$-colorability.
For each vertex $v$ of $G$, add $M := (n+1)(c+1)$ vertices $u_v$ to $G$ and connect them to $v$.
For each edge $vw$ of $G$, add $M$ vertices $u_{vw}$ to $G$ and connect them to both $v$ and $w$.
Let $G'$ be the resulting graph.
Clearly, the size of $G'$ is polynomial in the size of $G$.
Additionally, $G'$ is planar if $G$ is, and $G'$ has a conflict-free $k$-coloring of size $n$ iff $G$ is properly $k$-colorable:
Any proper $k$-coloring of $G$ is a conflict-free $k$-coloring of $G'$, as every vertex added to $G$ is either adjacent to two distinctly colored vertices of $G$, or adjacent to just one vertex of $G$.
Conversely, let $\chi$ be a conflict-free coloring of $G'$, coloring just $n$ vertices.
If $\chi$ did not assign a color to some vertex $v$ of $G$, it would have to color all $M \geq n+1$ neighbors of $v$.
If $\chi$ assigned the same color to any pair $v,w$ of vertices adjacent in $G$, it would have to color all $M$ vertices adjacent only to $v$ and $w$.
Therefore, $\chi$ is a proper coloring of $G$.
\revised{Running a $c$-approximation algorithm $\mathcal{A}$ for $\gamma^k_{CF}$ on $G'$ results in an approximate value $\mathcal{A}(G') \leq c \cdot \gamma^k_{CF}(G')$.
We have $\mathcal{A}(G') \leq c \cdot n < M$ if $G$ is $k$-colorable, and $\mathcal{A}(G') \geq M$ if $G$ is not; thus we could decide proper $k$-colorability in polynomial time.}
%Running a $c$-approximation algorithm on $G'$ results in a conflict-free coloring using at most $c \cdot n < M$ vertices if $G$ is $k$-colorable, and using at least $M$ vertices if $G$ is not; thus we could decide $k$-colorability in polynomial time.
\end{proof}
\end{theorem}

% !TEX root = ./main.tex
\section{Closed Neighborhoods: Planar Conflict-Free Coloring}
\label{sec:conflict-free-coloring-planar}
This section deals with the {\sc{Planar Conflict-Free $k$-Coloring}} problem which consists of deciding conflict-free $k$-colorability for fixed $k$ on planar graphs.
Due to the 4-color theorem, we immediately know that every planar graph is conflict-free 4-colorable.
This naturally leads to the question of whether there are planar graphs requiring 4 colors or whether fewer colors might already suffice for a conflict-free coloring, which we address in the following two sections.

\subsection{Complexity}
For $k \in \{1,2\}$ colors, we show that the problem of deciding conflict-free $k$-colorability on planar graphs is NP-complete. This implies that 2 colors are not sufficient.

%For $k=1$, deciding planar conflict-free 1-colorability is equivalent to deciding whether a planar graph admits a dominating set of vertices with pairwise distance at least three.
\begin{theorem}
\label{thm:planar-1-coloring-npc}
Deciding planar conflict-free {1-color\-abi\-lity} is NP-complete.
\begin{proof}
Membership in NP is obvious.
The proof of NP-hardness is done by reduction from the problem {\sc Positive Planar 1-in-3-SAT}.
From a positive planar 3-CNF formula $\phi$ with clauses $C = \{c_1,\ldots,c_l\}$ and variables $X = \{x_1,\ldots,x_n\}$ we construct in polynomial time a graph $G_1(\phi)$ such that $\phi$ is 1-in-3-satisfiable iff $G_1(\phi)$ admits a conflict-free 1-coloring.

First, find and fix a planar embedding $d$ of $G(\phi)$.
$G_1(\phi)$ is constructed from $G(\phi)$ and $d$ as follows:
For every variable $x_i$, there is a cycle $Z_i = (z_{i,1},\ldots,z_{i,12})$ of length $12$.
The vertices $z_{i,1},z_{i,4},z_{i,7},z_{i,10}$ are referred to as \emph{true} vertices of $Z_i$, all other vertices are \emph{false} vertices.
Moreover, vertices $z_{i,1},z_{i,2},z_{i,3}$ are called \emph{upper} vertices of $Z_i$, and vertices $z_{i,7},z_{i,8},z_{i,9}$ are called \emph{lower} vertices of $Z_i$.
Additionally, vertices $z_{i,4},z_{i,5},z_{i,6}$ are called \emph{right} vertices of $Z_i$ and $z_{i,10},z_{i,11},z_{i,12}$ are called \emph{left} vertices of $Z_i$.

For each clause $c_j$, there is a cycle $(c_{j,1},\ldots,c_{j,4})$ of length 4 in $G_1(\phi)$.
To each variable $x_i$ for $i \in \{2,\ldots,n-1\}$, we associate two disjoint sequences $U_i = \big(u_j\big)_{j=1}^{|U_i|}$ and $L_i = \big(l_j\big)_{j=1}^{|L_i|}$ of clauses $x_i$ appears in.
The sequences are constructed using a clockwise (with respect to $d$) enumeration of the edges of $x_i$ in $G(\phi)$, starting with $x_{i-1}x_i$.
Let $(x_{i-1}x_i,x_ic_{j_1},\ldots,x_ic_{j_{\lambda}},x_ix_{i+1},x_ic_{j_{\lambda+1}},\ldots,x_ic_{j_{\mu}})$ be the sequence of edges encountered in this manner and set $U_i := (c_{j_1},\ldots,c_{j_{\lambda}})$ and $L_i := (c_{j_{\lambda+1}},\ldots,c_{j_{\mu}})$.
For $i \in \{1,n\}$, $L_i$ is empty and $U_i$ contains all clauses $x_i$ appears in, again in clockwise order.
In $G_1(\phi)$, the clauses and variables are connected such that for each clause $c_j$ that $x_i$ occurs in, either the upper or the lower \emph{true} vertex of $x_i$ is adjacent to $c_{j,1}$.
More precisely, for variable $x_i$, if $c_j = u_m$, we add the edge $c_{j,1}z_{i,1}$ to connect the upper true vertex to the clause.
If $c_j = l_m$, we add $c_{j,1}z_{i,7}$ to connect the lower true vertex to the clause.
%In this way, only true vertices are connected to a clause; false vertices only have their two neighbors on the cycle $Z_i$.
Because the order of edges around each vertex is preserved by the construction, the graph $G_1(\phi)$ obtained in this way can be embedded in the plane by a suitable adaptation of $d$. See Figure~\ref{fig:figurereduction} for an example of the construction.

\begin{figure}[!htb]
\centering
\includegraphics[width=8cm]{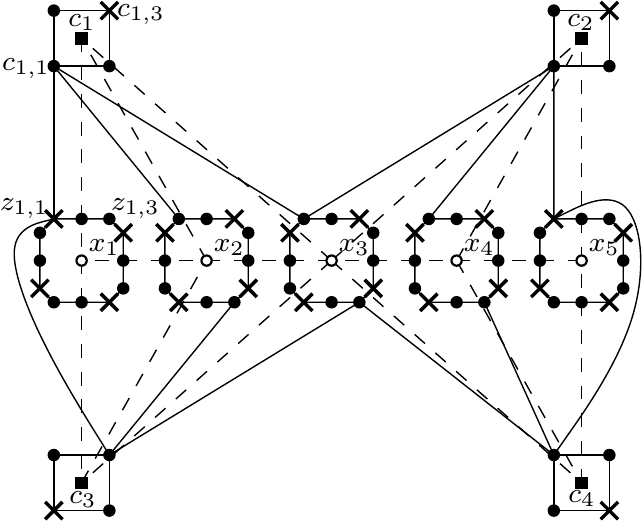}
\caption{A formula graph $G(\phi)$ (dashed) and the corresponding $G_1(\phi)$ (solid).}
\label{fig:figurereduction}
\end{figure}

Now we prove that $G_1(\phi)$ is conflict-free 1-colorable iff $\phi$ is 1-in-3-satisfiable.
Regarding necessity, a valid truth assignment $b: X \to \mathbb{B}$ yields a valid conflict-free coloring by coloring the vertex $c_{j,3}$ of every clause, coloring all true vertices of variables with $b(x_i) = 1$ and coloring the false vertices $z_{i,3},z_{i,6},z_{i,9},z_{i,12}$ of all other variables. Thus, in every cycle $Z_i$, every third vertex is colored, providing a conflict-free neighbor to every vertex of $Z_{i}$.
Moreover, in each clause, by virtue of $c_{j,3}$ being colored, vertices $c_{j,2},c_{j,3},c_{j,4}$ have a conflict-free neighbor.
Because $b$ is a valid truth assignment, for each clause, the vertex $c_{j,1}$ is adjacent to exactly one colored true vertex.
Therefore, the coloring constructed in this way is conflict-free.

Regarding sufficiency, we first argue that the vertices $c_{j,1},c_{j,2},c_{j,4}$ can never be colored:
If $c_{j,1}$ receives a color, then $c_{j,3}$ still enforces that one of $c_{j,2},c_{j,3},c_{j,4}$ is colored, leading to a contradiction in either case.
If $c_{j,2}$ receives a color, then $c_{j,4}$ cannot have a conflict-free neighbor and vice versa.
Therefore, no clause vertex can be the conflict-free neighbor of any vertex of $Z_{i}$. Thus, the conflict-free neighbor of \emph{every} vertex of $Z_{i}$ must itself be a vertex of $Z_{i}$. Moreover, the conflict-free neighbor of every vertex $c_{j,1}$ must be a true vertex.
Thus, there are exactly three ways to color each cycle $Z_{i}$: either by coloring the true vertices (one possibility), or by coloring every other false vertex (two possibilities).
A valid conflict-free 1-coloring of $G_1(\phi)$ satisfies the property that for each clause $c_j$, exactly one of the true vertices adjacent to $c_{j,1}$ is colored.
Hence, a valid conflict-free 1-coloring of $G_1(\phi)$ induces a valid truth assignment $b$ by setting $b(x_i) = 1$ iff all true vertices of $x_i$ are colored.
\end{proof}
\end{theorem}

%\subsection{Two colors}
%We use the graph $G_1(\phi)$ in conjunction with a gadget to prove the following theorem for $k=2$ colors:
\begin{theorem}\label{thm:planar-2-coloring-npc}
It is NP-complete to decide whether a planar graph admits a conflict-free 2-coloring.
\end{theorem}
\noindent The proof requires the gadget $G_{\leq 1}$ depicted in Figure~\ref{fig:cf-2-coloring-leq1-gadget}.
$G_{\leq 1}$ consists of three vertices $v,w_1,w_2$ forming a triangle.
Each edge $ux$ of the triangle has two corresponding vertices $y^1_{ux},y^2_{ux}$, each connected to $u$ and $x$.
Furthermore, both $w_1$ and $w_2$ are attached to two copies of a cycle on 4 vertices, where every vertex of both cycles is adjacent to the corresponding $w_i$.
$G_{\leq 1}$ can be used to enforce that the vertices connected to its central vertex $v$ are colored using at most one distinct color:

\begin{lemma}
\label{lem:cf-2-coloring-leq1-gadget}
Let $G = (V,E)$ be any graph, let $v \in V$ and let $G'$ be the graph resulting from adding a copy of $G_{\leq 1}$ to $G$ by identifying $v$ in $G$ with $v$ in $G_{\leq 1}$.
Then (1) $G'$ is planar if $G$ is, and (2) every conflict-free 2-coloring of $G'$ leaves $v$ uncolored and uses at most one color on~$N_G[v]$.
\begin{proof}
The planarity of $G'$ follows from the planarity of $G$ by the observation that $G_{\leq 1}$ is planar and can be embedded in any face incident to $v$ in a planar embedding of $G$.
Now consider a conflict-free 2-coloring $\chi$ of $G'$.
$\chi$ must color both $w_1$ and $w_2$.
Otherwise, $\chi$ restricted to each of the two 4-cycles adjacent to $w_i$ must be a valid conflict-free 2-coloring.
However, as $C_4$ requires at least 2 different colors, $w_i$ then sees two occurrences of both colors, and thus cannot have a  conflict-free neighbor anymore.
Furthermore, $\chi(w_1) \neq \chi(w_2)$, as otherwise, $y^1_{w_1w_2}$ and $y^2_{w_1w_2}$ must both be colored with the other color; but then, $w_1$ and $w_2$ again see two occurrences of both colors.
By an analogous argument, $\chi$ must not color $v$. Moreover, $\chi$ cannot use more than one color on $N_G[v]$, because $v$ already sees one occurrence of each color, so adding another occurrence of both colors would yield a conflict at $v$.
\end{proof}
\end{lemma}
\begin{figure}[!htb]
\minipage{0.5\textwidth}
\centering
\includegraphics[width=6cm,height=6cm,keepaspectratio]{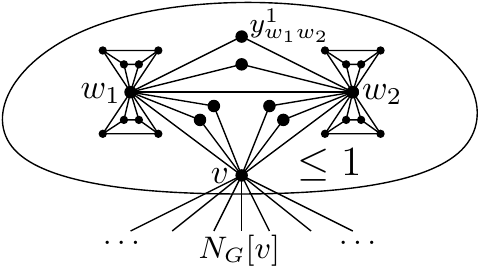}
\caption{Gadget $G_{\leq 1}$}
\label{fig:cf-2-coloring-leq1-gadget}
\vspace*{5mm}
\endminipage\hfill
\minipage{0.5\textwidth}
\centering
\includegraphics[width=8cm,height=8cm,keepaspectratio]{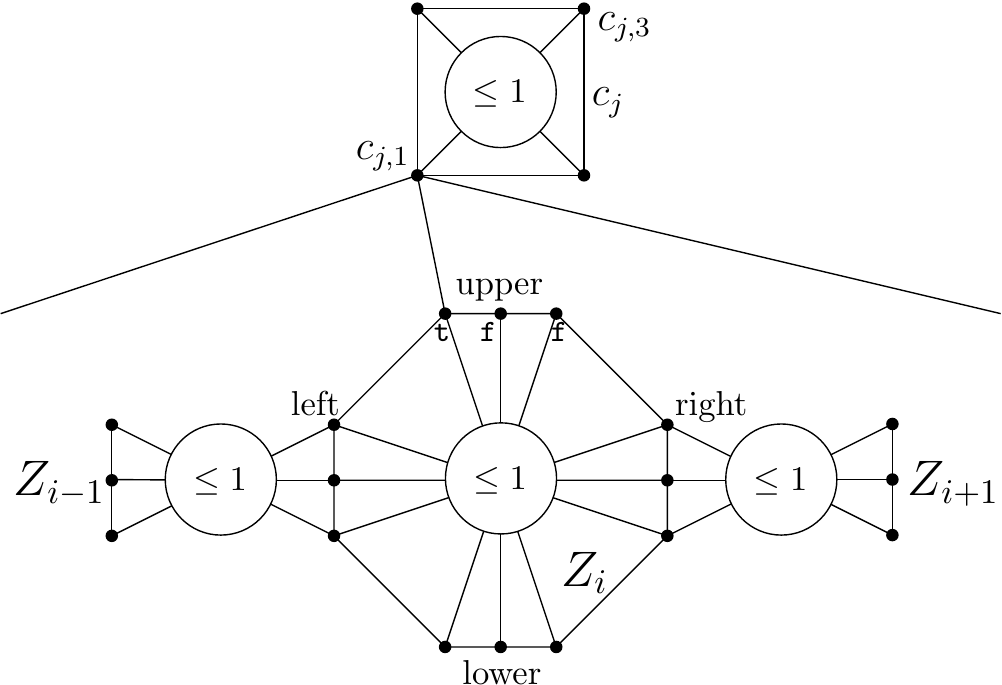}
\caption{Clause and variable gadget for $k=2$}
\label{fig:cf-2-coloring-reduction}
\endminipage\hfill
\end{figure}
\begin{proof}[Proof of Theorem~\ref{thm:planar-2-coloring-npc}]
NP-hardness is proven by constructing, in polynomial time, a planar graph $G_2(\phi)$ from the graph $G_1(\phi)$ used in the hardness proof for $k=1$, such that $G_2(\phi)$ is conflict-free 2-colorable iff $G_1(\phi)$ is conflict-free 1-colorable.

The construction is carried out by adding a gadget $G_{\leq 1}$ to every variable cycle $Z_i$ of $G_1(\phi)$, to every clause cycle and between the right and left vertices of two adjacent variable cycles $Z_i$ and $Z_{i+1}$.
This is depicted in Figure~\ref{fig:cf-2-coloring-reduction}.
More precisely, for every cycle $Z_i$, we add one copy of gadget $G_{\leq 1}$, and connect its central vertex $v$ to all vertices of the cycle.
In a planar embedding of $G_2(\phi)$, these gadgets can be embedded within the face defined by the cycles $Z_i$ and thus do not harm planarity.
By Lemma~\ref{lem:cf-2-coloring-leq1-gadget}, this enforces that on every cycle, only one color can be used.
Moreover, for every edge $x_ix_{i+1}$ in $G(\phi)$, we add one copy of $G_{\leq 1}$ that we connect to the right vertices of $x_i$ and the left vertices of $x_{i+1}$.
This preserves planarity because these gadgets and the added edges can be embedded in the face crossed by $x_ix_{i+1}$ in some fixed embedding $d$ of $G(\phi)$.
As one of the right vertices of $x_i$ and one of the left vertices of $x_{i+1}$ must be colored, this enforces that the same single color must be used to color all cycles $Z_i$.
Finally, we add a copy of $G_{\leq 1}$ to every clause $c_j$ and connect it to $c_{j,1},\ldots,c_{j,4}$.
Again, this preserves planarity because the gadget may be embedded in the face defined by $(c_{j,1},\ldots,c_{j,4})$.

We now argue that $G_2(\phi)$ is conflict-free 2-colorable iff $G_1(\phi)$ is conflict-free 1-colorable.
A 1-coloring of $G_1(\phi)$ induces a 2-coloring of $G_2(\phi)$ by copying the color assignment and coloring the internal vertices of the added gadgets as described in the proof of Lemma~\ref{lem:cf-2-coloring-leq1-gadget}.
Now, let $G_2(\phi)$ be conflict-free 2-colorable and fix a valid 2-coloring $\chi$.
In each clause, $\chi$ must color $c_{j,3}$ and neither of $c_{j,1},c_{j,2}$ nor $c_{j,4}$ can be colored.
Therefore, no clause vertex can be the conflict-free neighbor of any vertex of $Z_{i}$. Thus, the conflict-free neighbor of \emph{every} vertex of $Z_{i}$ must itself be a vertex of $Z_{i}$. Moreover, the conflict-free neighbor of every vertex $c_{j,1}$ must be a true vertex.
As there is only one color available to color all cycle vertices of all variables, the restriction of $\chi$ to the vertices of $G_1(\phi)$ yields a valid 1-coloring except for the fact that some $c_{j,3}$ might use a different color than the one used for the variables.
However, this can be fixed by simply replacing all occurring colors with one single color.
Hence, $G_2(\phi)$ is conflict-free 2-colorable iff $G_1(\phi)$ is conflict-free 1-colorable.
\end{proof}
%\subsubsection{Outerplanar graphs}
%The NP-completeness for two colors immediately implies that there are planar graphs requiring at least three colors.
%However, on the more restricted class of outerplanar graphs, two colors are already sufficient:

%\begin{proof}
%The class of outerplanar graphs contains exactly those graphs which have neither a $K_4$ nor a $K_{2,3}$ as a minor.
%Excluding a $K_{2,3}$ as a minor also excludes its supergraph $K_5 \setminus \Delta$.
%Therefore, Theorem~\ref{thm:sufficient-criterion-excluded-minors} applies, implying that every outerplanar graph admits a conflict-free 2-coloring which can be computed in polynomial time.
%Furthermore, Theorem~\ref{thm:DPoptimalconflictfreecoloring} implies a 2-coloring with minimal number of colored vertices can be computed efficiently.
%\end{proof}

\subsection{Sufficient Number of Colors}\label{sec:conflict-free-coloring-planar-sufficiency}
%For three colors, the situation is different: We prove that all planar graphs are conflict-free 3-colorable and that a corresponding coloring can be computed in polynomial time.
As shown above, it is NP-complete to decide whether a planar graph has a conflict-free $k$-coloring for $k \in \{1,2\}$. On the positive side, we can establish
the following result, which follows from the more general results discussed in \S~\ref{sec:general-graphs-sufficiency}.

\begin{corollary}[of Theorem \ref{thm:sufficient-criterion-excluded-minors}]\label{thm:sufficiency}
Every outerplanar graph is conflict-free $2$-colorable and every planar graph is conflict-free $3$-colorable. Moreover, such colorings can be computed in polynomial time.
\end{corollary}

Outerplanar graphs are not the only interesting graph class for which one might
suspect two  colors to be sufficient. Two other interesting subclasses of
planar graphs are series-parallel graphs and pseudomaximal planar graphs.
However, each of these classes contains graphs that do not admit a
conflict-free 2-coloring: The graph $G_3$ as defined in
\S~\ref{sec:conflict-free-coloring-of-general-graphs} is an example of a
series-parallel graph requiring three colors.  Figure~\ref{fig:max-op-cf2}
depicts a maximal outerplanar graph $O_9$ satisfying $\chi_{CF}(O_9) = 2$.
This graph can be used to obtain a pseudomaximal planar graph $M$ with $\chi_{CF}(M) = 3$ by adding two copies of $O_9$ to the neighborhood of every vertex of a triangle, similar to the construction of $G_3$, and adding gadgets on the inside of the triangle as depicted in Figure~\ref{fig:pseudomax-planar-cf3}.
\begin{figure}[h]
\centering
\includegraphics[width=4cm]{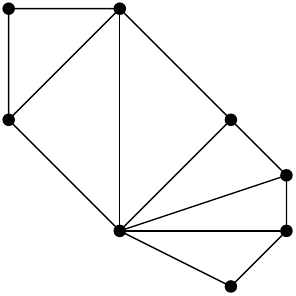}
\caption{The maximal outerplanar graph $O_9$.}
\label{fig:max-op-cf2}
\end{figure}
\begin{figure}[hbt]
\centering
\includegraphics[width=.5\linewidth]{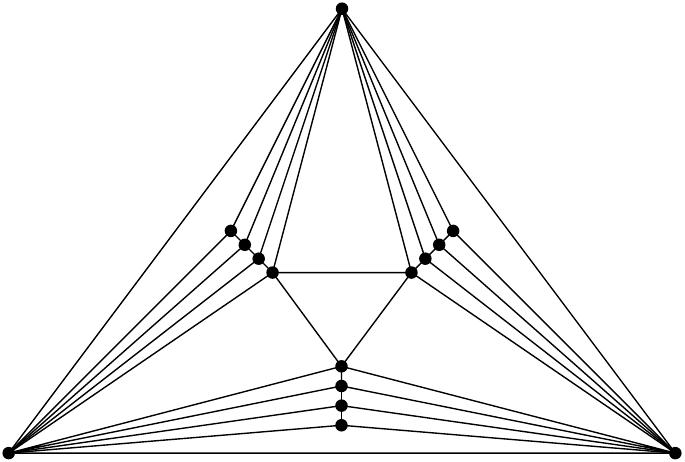}
\caption{The pseudomaximal planar graph $M$, without the $O_9$ gadgets.}
\label{fig:pseudomax-planar-cf3}
\end{figure}

Furthermore, observe that Theorem \ref{thm:sufficiency} does not hold if every vertex must be colored.
In this case, there are outerplanar graphs requiring 3 colors for a conflict-free coloring.
One can obtain an example of such a graph by adding a chord to a cycle of length 5.

%\begin{proof}
%By Wagner's theorem, the class of planar graphs contains exactly those graphs which have neither a $K_5$ nor a $K_{3,3}$ as a minor.
%Excluding $K_{3,3}$ as a minor also excludes $K_6 \setminus \Delta$, because $K_{3,3}$ is a subgraph of $K_6 \setminus \Delta$.
%Therefore, Theorem~\ref{thm:sufficient-criterion-excluded-minors} applies, implying that every planar graph admits a conflict-free 3-coloring which can be computed in polynomial time.
%\end{proof}

% !TEX root = ./main.tex
\section{Closed Neighborhoods: {Planar Conflict-Free Domination}}
\label{sec:bicriteria-conflict-free-coloring-planar}
%In this section we consider the problem of minimizing the number of colored vertices in a conflict-free $k$-coloring for a fixed $k$, which is equivalent to computing $\gamma^{k}_{CF}$. 
In this section we consider the decision problem {\sc{$k$-Conflict-Free Dominating Set}} for planar graphs. 
In \S~\ref{sec:bicriteria-at-most-two-colors}, we deal with the cases when $k \in \{1, 2 \}$ for planar and outerplanar graphs, and we give a polynomial time algorithm to compute an optimal conflict-free coloring of outerplanar graphs with $k \in \{1, 2 \}$ colors. Section~\ref{three_colors} discusses the problem for~\mbox{$k\geq 3$.}

\subsection{At Most Two Colors}\label{sec:bicriteria-at-most-two-colors}
We start by pointing out that, for every conflict-free $1$-colorable graph $G$, $\gamma^{1}_{CF}(G) = \gamma(G)$ holds. Moreover, Corollary~\ref{cor:complexity-12-dominating-set} discusses the complexity of {\sc{$k$-Conflict-Free Dominating Set}} and Theorem~\ref{thm:DPoptimalconflictfreecoloring} states positive results for outerplanar graphs.

\begin{corollary}[of Theorems~\ref{thm:planar-1-coloring-npc} and~\ref{thm:planar-2-coloring-npc}]\label{cor:complexity-12-dominating-set}
\ \\
{\sc{$k$-Con\-flict-Free Dominating Set}} is NP-complete for $k \in \{1, 2 \}$ for planar graphs.
\end{corollary}

%\begin{proof}
%Membership in NP is clear.
%The NP-hardness for $k \in \{1,2\}$ in planar graphs follows from Theorems~\ref{thm:planar-1-coloring-npc} and~\ref{thm:planar-2-coloring-npc}. 
%%the number of colored vertices in case of conflict-free $1$-colorable is always equal to the domination number $\gamma(G)$. 
%\end{proof}

%\noindent On the positive side, we can show the following result:

\begin{restatable}{theorem}{outerplanar}\label{thm:DPoptimalconflictfreecoloring}
	Let $k\in\{1,2\}$ and let $G$ be an outerplanar graph. We can decide in polynomial time whether $\chi_{CF}(G)\leq k$. Moreover, we can compute a conflict-free $k$-coloring of $G$ that minimizes the number of colored vertices in \revised{$\mathcal{O}(n^{4k+1})$} time.
\end{restatable}

The proof of Theorem~\ref{thm:DPoptimalconflictfreecoloring} relies on a polynomial-time algorithm that computes a $k$-coloring of the input outerplanar graph $G$ if and only if such a coloring exists.
\revised{Intuitively speaking, our algorithm works as follows.
For each vertex $v \in G$ and each edge $vw \in G$, we consider all possible assignments of conflict-free neighbors to $v$ and $w$ and colors to these conflict-free neighbors.
Each such assignment is called a \emph{neighborhood configuration}.
Because the number of colors is constant and there is at most one conflict-free neighbor per color for each vertex, there are only polynomially many neighborhood configurations for each vertex or edge.

\begin{figure}[h]
	\begin{center}
		\begin{subfigure}[b]{.42\linewidth}
			\resizebox{\linewidth}{!}{\includegraphics{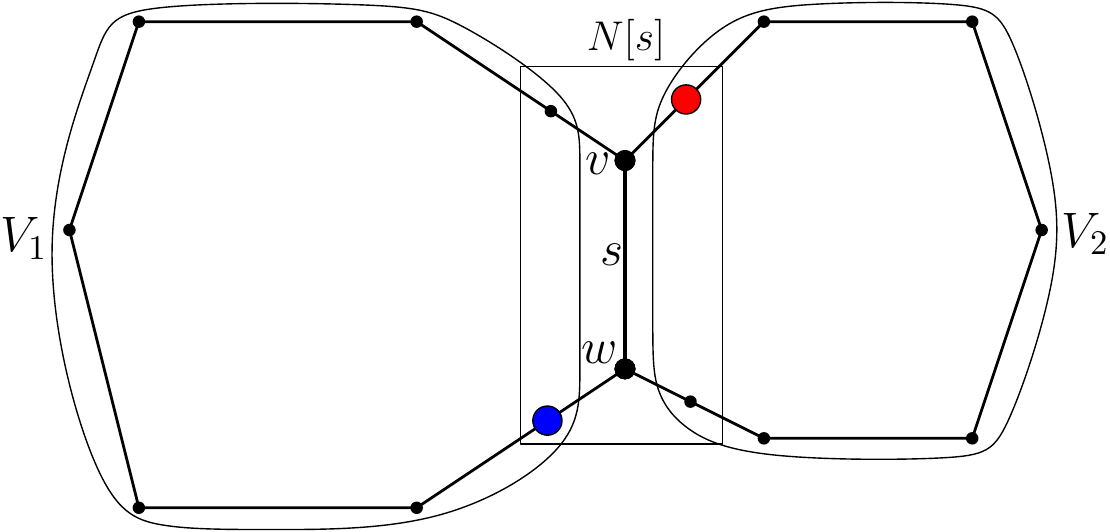}}
			\caption{An outerplanar graph $G$ and an edge separator $s = \{v,w\}$ splitting $G$ into components $V_1$ and $V_2$ with a neighborhood configuration of $s$ (colored vertices).}
		\end{subfigure}
		\hfill
		\begin{subfigure}[b]{.54\linewidth}
			\resizebox{\linewidth}{!}{\includegraphics{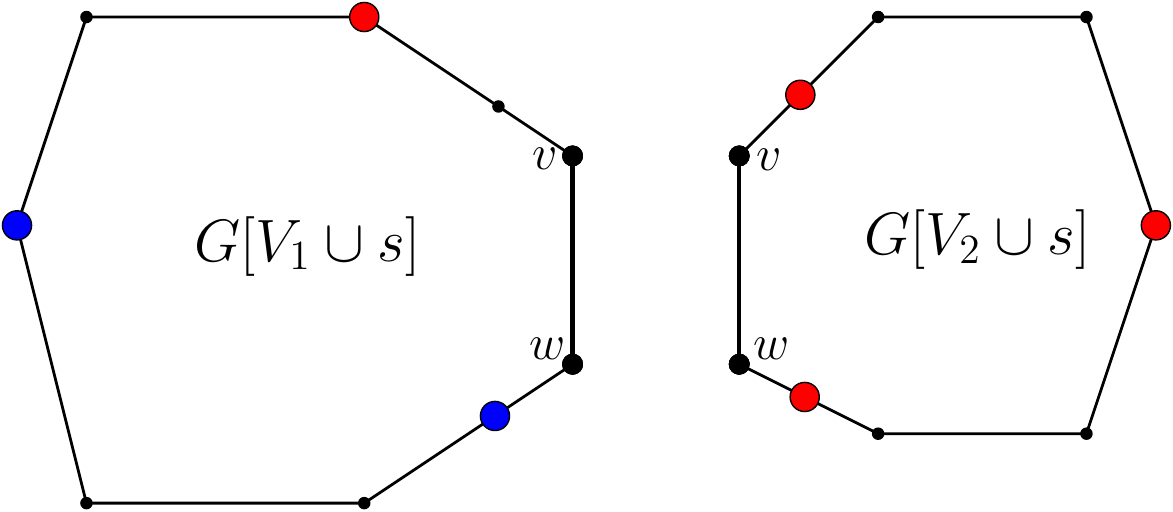}}
			\caption{Colorings extending the neighborhood configuration of $s$. These colorings are conflict-free on $V_1$ and $V_2$ and can be combined to a conflict-free coloring of $G$. Note that $u \notin V_1$ does not need to have a conflict-free neighbor in the coloring of $V_1$.}
		\end{subfigure}
	\end{center}
	\caption{If we fix a neighborhood configuration of a separator of $G$ and find conflict-free colorings of the separated components that extend this neighborhood configuration, we can combine these colorings to a conflict-free coloring of $G$.}
	\label{fig:example_combination_colorings}
\end{figure}

We decompose the outerplanar input graph at vertex separators (articulation points) and edge separators (edges shared by faces); removing the vertices of a separator splits the graph into several components.
The following key property of this decomposition is the basis for our dynamic programming algorithm; see Figure~\ref{fig:example_combination_colorings}.
Let $s$ be a separator in our graph, and let $V_1,\ldots,V_k$ be the vertex sets of the components of $G$ after removing $s$.
If we fix a neighborhood configuration $\mathcal{C}$ of $s$ and find, for each component $G[V_i \cup s]$, a coloring extending $\mathcal{C}$ that is conflict-free on $V_i$, then we can combine these colorings to a coloring of $G$.

\begin{figure}[h!]
	\begin{center}
		\resizebox{.7\textwidth}{!}{\includegraphics{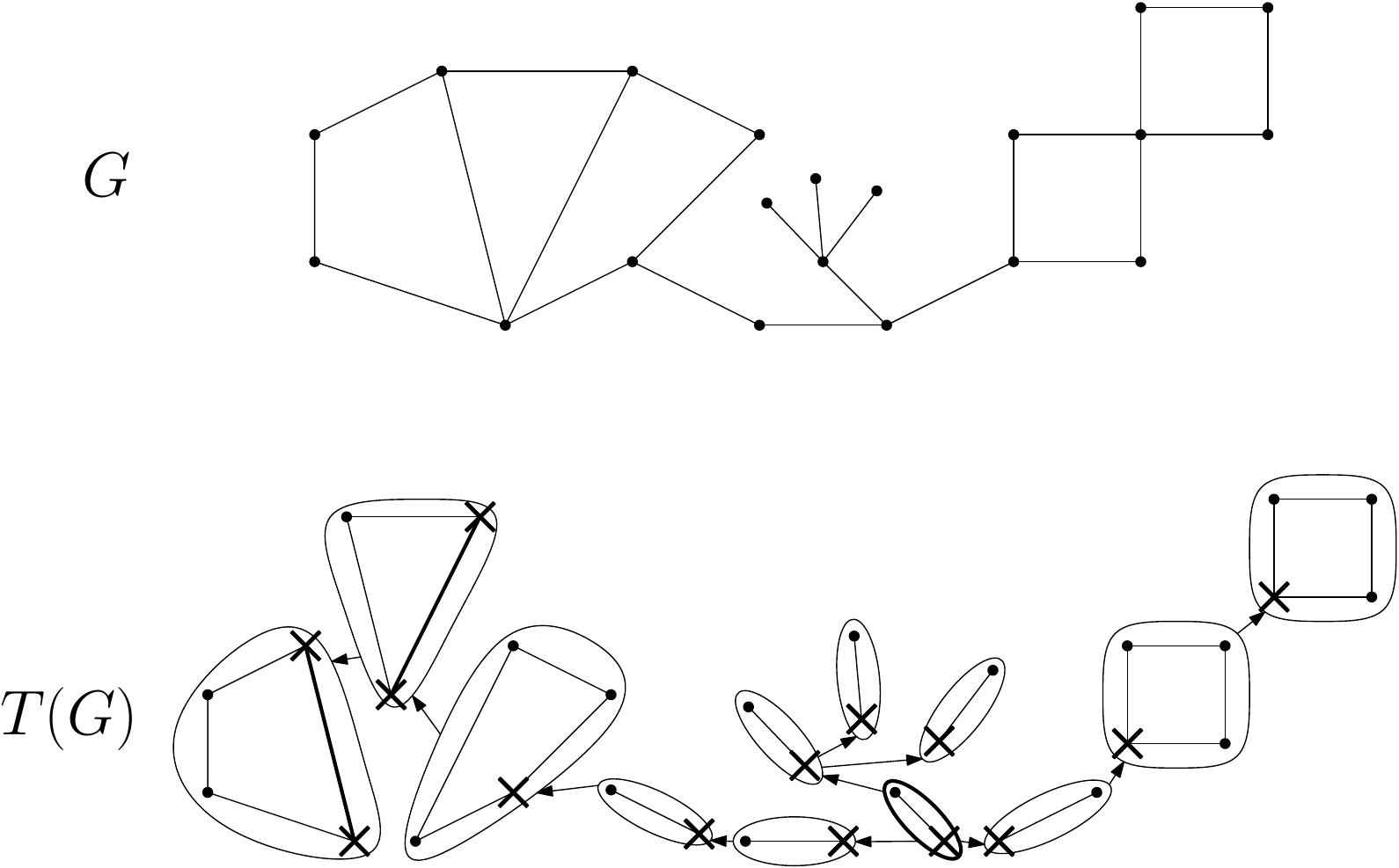}}
	\end{center}
	\caption{The arborescence $T(G)$ (bottom) for an outerplanar graph $G$ (top). The vertices of the incoming separator of each atom are marked ${\bf \times}$; incoming edge separators are drawn in bold. The root atom is drawn with bold outline and has an arbitrarily chosen vertex as incoming separator.}
	\label{fig:op-decomposition}
\end{figure}

Our decomposition yields an arborescence of components as depicted in Figure~\ref{fig:op-decomposition}, with edges between components that share a separator, using an arbitrary component as root and directing all edges accordingly.
In this arborescence, each component except the root has a unique incoming edge corresponding to a separator, called the \emph{incoming separator} of the component marked ${\bf \times}$ in Figure~\ref{fig:op-decomposition}.
Starting at the leaves, we use dynamic programming on this arborescence as follows.
For each component and each possible neighborhood configuration of the incoming separator, we compute a conflict-free $k$-coloring that extends the neighborhood configuration and minimizes the number of colored vertices, or find that this neighborhood configuration does not allow a conflict-free $k$-coloring.
At the root, this allows us to determine whether the graph is conflict-free $k$-colorable.
Moreover, if the graph is colorable, we can retrieve a coloring that minimizes the number of colored vertices.
In the following, we give a detailed formal description of this algorithm, prove its correctness and analyze its runtime.
}

\subsubsection{Preliminaries}
\revised{Let $G = (V, E)$ be an outerplanar graph.
W.l.o.g., we assume that $G$ is connected and has at least two vertices.
Let $\chi: V^{\prime}\subseteq V(G)\to\{0, 1,\ldots, k\}$ be a partial coloring of the vertices of $G$ and let $v\in V$.
Observe that $\chi$ defined like this modifies the definition given in the introduction by assigning color $0$ to uncolored vertices.
We begin by defining \emph{vertex neighborhood configurations}.
Intuitively speaking, a vertex neighborhood configuration assigns a color to $v$ and lists all conflict-free neighbors of $v$ together with their color, see Figure~\ref{fig:vertexconfiguration}.

\begin{figure}[ht]
\begin{center}
	\begin{subfigure}[b]{.29\linewidth}
		\resizebox{\linewidth}{!}{\includegraphics{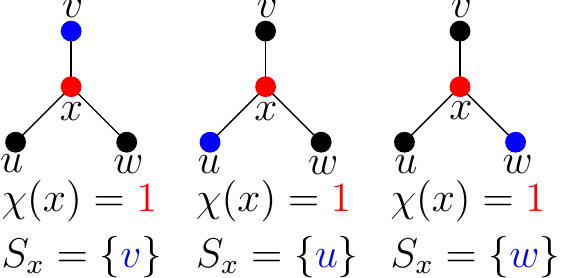}}
		\caption{$x$ is colored red and has one conflict-free neighbor.}
	\end{subfigure}
	\hspace{.03\linewidth}
	\begin{subfigure}[b]{.29\linewidth}
		\resizebox{\linewidth}{!}{\includegraphics{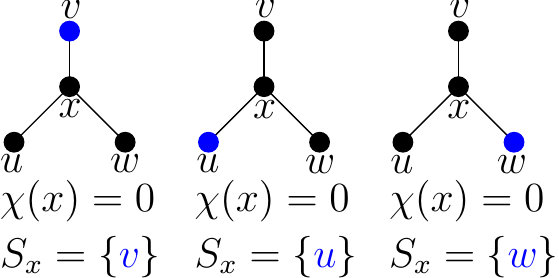}}
		\caption{$x$ is uncolored and has one conflict-free neighbor.}
	\end{subfigure}
	\hspace{.03\linewidth}
	\begin{subfigure}[b]{.32\linewidth}
		\resizebox{\linewidth}{!}{\includegraphics{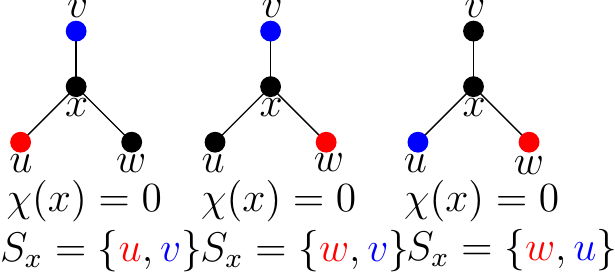}}
		\caption{$x$ is uncolored and has two conflict-free neighbors.}
	\end{subfigure}
\end{center}
\vspace*{-12pt}
\caption{A vertex $x$ with three neighbors and the possible neighborhood configurations of $s$, modulo switching labels of the colors.}
\label{fig:vertexconfiguration}
\end{figure}

\begin{Definition}[Vertex neighborhood configuration]\label{def:confvertex}
	A \emph{vertex neighborhood configuration} is a tuple $\mathcal{C}_{v} = [\chi(v), S_{v}, \rho_v]$, where $\chi(v) \in \{0,1,\ldots, k\}$ denotes the color of $v$; if $\chi(v) = 0$, we regard $v$ as uncolored, see Figure~\ref{fig:vertexconfiguration}.
	The set $\emptyset \neq S_v \subseteq N[v]$ contains all conflict-free neighbors of $v$.
	Because there is at most one conflict-free neighbor for each color, $S_v$ contains at most $k$ elements.
	Finally, $\rho_v: S_v \to \{1,\ldots,k\}$ is an injective assignment of colors to the conflict-free neighbors of $v$ such that $v \in S_v$ implies $\chi(v) = \rho_v(v)$.
\end{Definition}
}
	\revised{
		We call two vertex neighborhood configurations $\mathcal{C}_u,\mathcal{C}_v$ for adjacent vertices $u$ and $v$ \emph{compatible} if they do not contradict each other in the following sense, see Figure~\ref{fig:edgeconfiguration}.
		Firstly, they must not assign different colors to the same vertex.
		Secondly, after combining the partial colorings induced by $\mathcal{C}_u$ and $\mathcal{C}_v$, all conflict-free neighbors specified in the neighborhood configurations must remain conflict-free.
		An edge neighborhood configuration consists of two compatible vertex neighborhood configurations for its endpoints.
		Formally, we define this as follows.
	}

\revised{
\begin{figure}[ht]
\begin{center}
	\begin{subfigure}[b]{.16\linewidth}
		\resizebox{\linewidth}{!}{\includegraphics{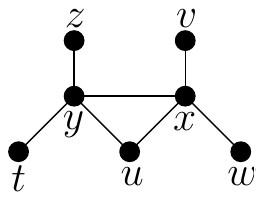}}
		\caption{The edge $xy$ and the neighbors of $x$ and $y$.}
	\end{subfigure}
	\hspace{.03\linewidth}
	\begin{subfigure}[b]{.67\linewidth}
		\resizebox{\linewidth}{!}{\includegraphics{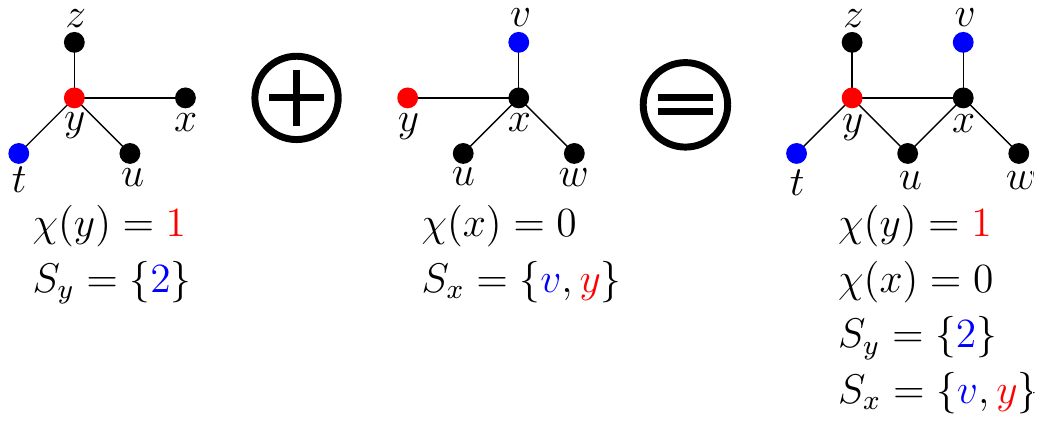}}
		\caption{The combination of two neighborhood configurations of $x$ and $y$.}
	\end{subfigure}
\end{center}
\vspace*{-12pt}
\caption{The combination of two neighborhood configurations of two adjacent vertices $x$ and $y$ results in a neighborhood configuration of the edge $xy$.}
\label{fig:edgeconfiguration}
\end{figure}

\begin{Definition}[Edge neighborhood configuration]\label{def:confedge}
	For an edge $uv$, we say that $\mathcal{C}_{u} = [\chi(u), S_{u}, \rho_u]$ and $\mathcal{C}_{v} = [\chi(v), S_{v}, \rho_v]$ are \emph{compatible}, denoted by $\mathcal{C}_{u}\leftrightarrow\mathcal{C}_{v}$, if the following conditions hold, see Figure~\ref{fig:edgeconfiguration}.
	\begin{enumerate}
		\item For every $w \in S_v \cap S_u$, $\rho_u(w) = \rho_v(w)$. If $u$ is in $S_v$, then $\chi(u)$ must be $\rho_v(u)$, and vice versa.
		\item The combined coloring
			$$\rho_{uv}: S_u \cup S_v \cup \{u,v\} \to \{0,\ldots,k\}, w \mapsto \begin{cases}\chi(w), &\mbox{if }w \in \{u,v\},\\
			\rho_u(w), &\mbox{if }w \in S_u,\\
			\rho_v(w) &\mbox{otherwise.}
			\end{cases}$$ must be injective on $N[v]$ and $N[u]$, with the exception that both $u$ and $v$ may receive color 0.
%		\item \revised{For each color, there is at most one conflict-free neighbor $w \in S_u \cup S_v$ in $N[u]$ and at most one conflict-free neighbor $w \in S_u \cup S_v$ in $N[v]$.}
%		\item \revised{There is no conflict-free neighbor $w \in S_v \setminus \{ u \}$ with color $\chi(u)$ and vice versa.}
	\end{enumerate}
	An \emph{edge neighborhood configuration} of $e = uv$ is a pair $\mathcal{C}_e = [\mathcal{C}_u, \mathcal{C}_v]$ of compatible vertex neighborhood configurations.
	For $w \in \{u,v\}$, $\mathcal{C}_e^w$ shall denote the neighborhood configuration of $w$ contained in $\mathcal{C}_e$.
\end{Definition}
}
%	Next, we give the definition of a configuration of an edge which basically the combination of two compatible configurations of the two vertices of the edge.
%\begin{Definition}
%	For $e = uv\in E$ we call a pair $\mathcal{C}_{e} = [\mathcal{C}_{u}, \mathcal{C}_{v}]$ of compatible configurations a \emph{configuration} of $e$.
%	By $\mathcal{C}_{e}^{w}$ we denote the configuration of an endpoint $w\in\{u,v\}$ of $e$. Observe that if $\chi$ was conflict-free, then $S_{v}\neq\emptyset$.
%\end{Definition}
%	Next, we define what it means that two configurations of two adjacent edges $uv$ and $vw$ are \emph{compatible} which is basically the case when the restrictions of both configurations to $v$ are equal.
\revised{
	Observe that we can check in $\mathcal{O}(k)$ time whether a pair of vertex neighborhood configurations is compatible.
	For a pair of incident edges, we call a pair of edge neighborhood configurations \emph{compatible} if the neighborhood configuration of $v$ is the same in both neighborhood configurations, see Figure~\ref{fig:compatible_edge_configurations}.
	
\begin{figure}[ht]
\begin{center}
	\begin{subfigure}[b]{.18\linewidth}
		\resizebox{\linewidth}{!}{\includegraphics{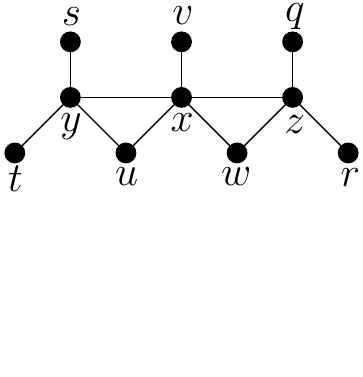}}
		\caption{Two adjacent edges $yx$ and $xz$ and the neighbors of $y$, $x$, and $z$.}
	\end{subfigure}
	\hspace{.03\linewidth}
	\begin{subfigure}[b]{.77\linewidth}
		\resizebox{\linewidth}{!}{\includegraphics{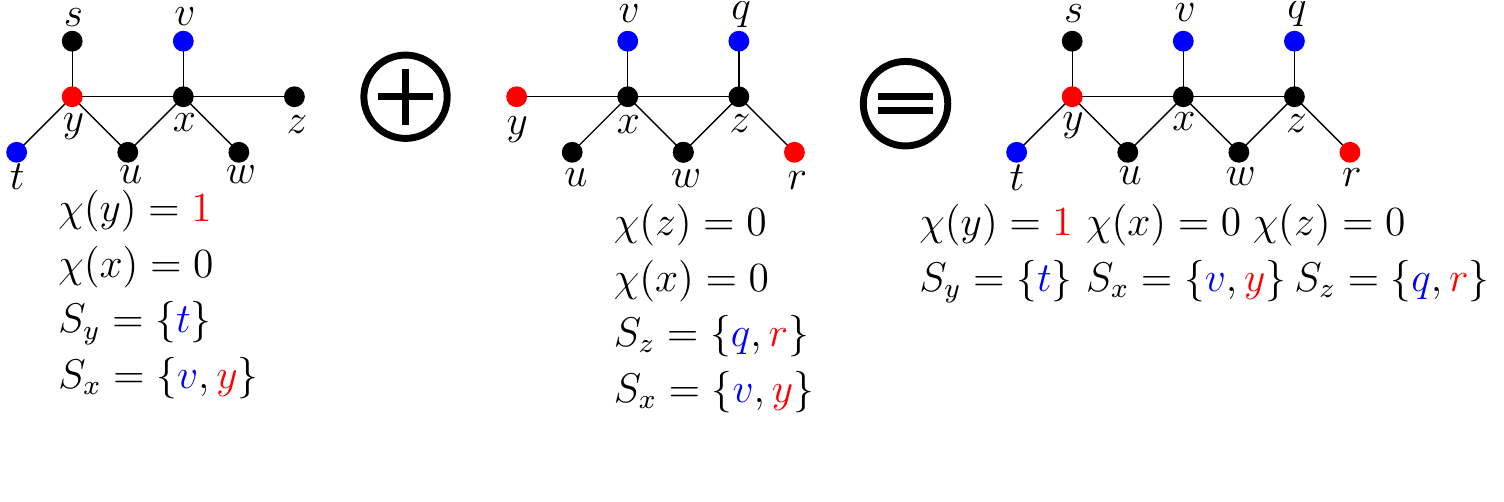}}
		\caption{Two compatible neighborhood configurations of the two adjacent edges $yx$ and $xz$.}
	\end{subfigure}
\end{center}
\vspace*{-12pt}
\caption{Compatible neighborhood configurations of adjacent edges.}
\label{fig:compatible_edge_configurations}
\end{figure}
	
\begin{Definition}[Compatibility]\label{def:compEdges}
	If $\mathcal{C}_{e}^{v} = \mathcal{C}_{e^{\prime}}^{v}$ for a pair $e = uv, e^{\prime} = vw$ of incident edges, then we say $\mathcal{C}_{e^{\prime}}$ \emph{is compatible with} $\mathcal{C}_{e}$, see Figure~\ref{fig:compatible_edge_configurations}.
\end{Definition}}

\revised{
	We observe that if we have a neighborhood configuration for each edge and all these neighborhood configurations are pairwise compatible, the colors of all vertices are fixed in a consistent manner and we can thus derive a conflict-free $k$-coloring from the neighborhood configurations.
\begin{observation}
\label{obs:coloring_from_configuration}
	Let $\mathcal{C}$ be a set of edge neighborhood configurations containing one neighborhood configuration $\mathcal{C}_e$ for each edge $e$.
	If $\mathcal{C}_{e}$ and $\mathcal{C}_{e^\prime}$ are compatible for \emph{every} pair $e = uv$, $e^{\prime} = vw$ of incident edges, a conflict-free $k$-coloring can be obtained from $\mathcal{C}$.  
\end{observation}}

\revised{
	Our algorithm works by dynamic programming on an arborescence $T(G)$ derived from a decomposition of $G$ along vertex separators and edge separators into components called \emph{atoms}.
	\begin{Definition}
		A \emph{vertex separator} of $G$ is an articulation point of $G$, i.e., a vertex whose removal disconnects $G$.
		An \emph{edge separator} of $G$ is an edge $uv$ of $G$ such that removing $u$ and $v$ disconnects $G$.
		An \emph{atom} of $G$ is either an edge atom (formed by an edge) or a face atom (induced chord-free cycle of $G$).
	\end{Definition}
	
	Observe that, because $G$ is outerplanar, any connected induced subgraph of $G$ with at least two vertices is either an atom or contains a separator.
	The vertex set $V(T(G))$ of the arborescence $T(G)$ consists of atoms of $G$ and is defined by induction on the induced subgraphs $G'$ of $G$ as follows.
	If an induced subgraph $G'$ of $G$ is an atom, $V(T(G')) = \{ G' \}$.
	If $G'$ is no atom, let $s = \{v\}$ be a vertex separator of $G'$ if one exists; otherwise, let $s = \{u,v\}$ be an edge separator of $G'$.
	Let $V'_1,\ldots,V'_\ell$ be the vertex sets of the connected components of $G' - s$, and let $G'_1,\ldots,G'_\ell$ be the subgraphs induced by $V'_1 \cup s,\ldots,V'_\ell \cup s$.
	Then $V(T(G')) = \bigcup_{1 \leq i \leq \ell} V(T(G'_i))$ is the set of all atoms obtained by further subdividing $G'$.

	There is an arc between two vertices of $T(G)$ if the two atoms share a separator.
	To avoid cycles in $T(G)$, if more than two atoms share a vertex separator, instead of introducing an arc between every pair of them, we pick an arbitrary atom among them and connect it to all other atoms sharing the vertex separator.
	Because $G$ is outerplanar, this yields a tree of atoms of $G$; we turn this tree into the arborescence $T(G)$ by picking an arbitrary root vertex and orienting all edges away from this root.
	Each vertex $a$ of $T(G)$ except for the root has a unique incoming arc corresponding to a unique separator, called the \emph{incoming separator} of the atom $a$.
	For the root atom $r$, we pick an arbitrary vertex of $r$ as incoming separator; in this way, each atom $a$ has exactly one incoming separator $s_a$.
	See Figure~\ref{fig:op-decomposition} for an example of the construction.
	For an atom $a \in T(G)$, we denote by $T(G,a)$ the subtree of $T(G)$ rooted at $a$.
	Moreover, let $S(G,a)$ be the subgraph of $G$ induced by all vertices occurring in any atom in $T(G,a)$.
}

\subsubsection{Description of the Algorithm}
	\revised{
	For each vertex and each edge, our algorithm keeps a list of \emph{feasible} neighborhood configurations.
	\begin{figure}[h!]
			\begin{center}
				\begin{subfigure}[b]{.27\linewidth}
					\resizebox{\linewidth}{!}{\includegraphics{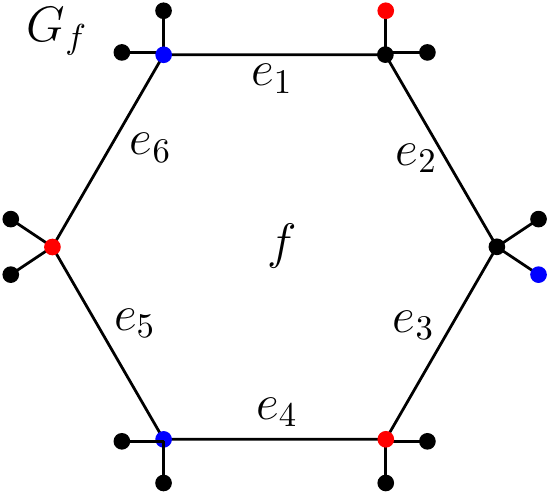}}
					\caption{A neighborhood configuration of a face $f$ of $G$.}
					\label{fig:dpfaceGraph}
				\end{subfigure}
				\hspace{.07\linewidth}
				\begin{subfigure}[b]{.55\linewidth}
					\resizebox{\linewidth}{!}{\includegraphics{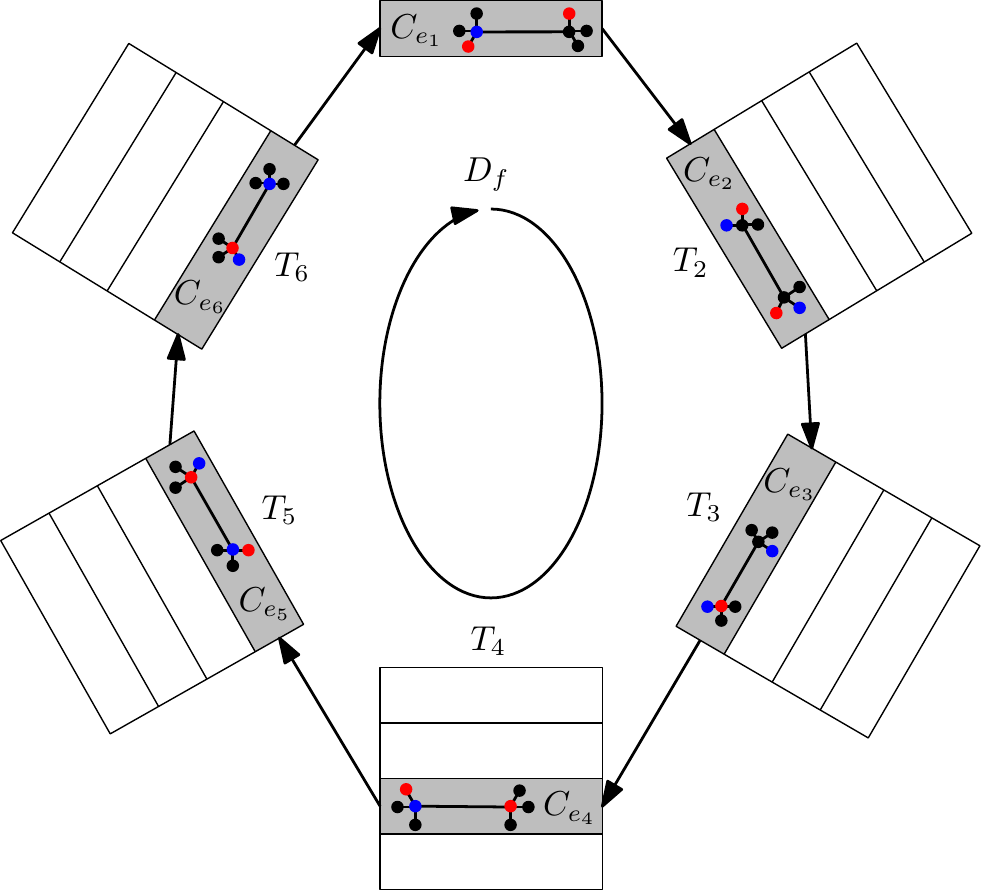}}
					\caption{The neighborhood configuration graph of $f$ and the cycle in the neighborhood configuration graph corresponding to the neighborhood configuration of Figure~\ref{fig:dpfaceGraph}.}
					\label{fig:dpfaceConfigurationGraph}
				\end{subfigure}
			\end{center}
			\caption{A face of $G$ and the corresponding neighborhood configuration graph.}
			\label{fig:dpface}
		\end{figure}%
		At any point in the algorithm, we know that any neighborhood configuration not on this list cannot be extended to a conflict-free $k$-coloring of $G$.
		Whenever we remove a neighborhood configuration from the list of feasible neighborhood configurations for a vertex, we also remove all corresponding neighborhood configurations from its incident edges.
		Similarly, when we remove the last neighborhood configuration of an edge that contains a certain vertex neighborhood configuration, we also remove that neighborhood configuration from the list of feasible neighborhood configurations of the vertex.
		In this way, deleting a feasible neighborhood configuration may cause a cascade of further deletions; however, a careful implementation of our algorithm can handle these deletions in $\mathcal{O}(1)$ time per deleted neighborhood configuration.
		Because each neighborhood configuration is deleted at most once, this does not affect our asymptotic running time.

		We initialize the lists of feasible neighborhood configurations by computing, for each vertex and each edge, the list of all possible neighborhood configurations according to Definitions~\ref{def:confvertex}~and~\ref{def:confedge}.
		We proceed by refining, for each atom $a \in V(T(G))$, the list of feasible neighborhood configurations of the incoming separator $s_a$.
		This process starts in the leaves of $T(G)$ and works its way up towards the root, terminating once the root has been processed.
		Processing an atom $a \in V(T(G))$ means removing all neighborhood configurations $\mathcal{C}_{s_a}$ of its incoming separator that cannot be extended to a conflict-free $k$-coloring of $S(G,a)$.
		Note that in this conflict-free $k$-coloring, the vertices of $s_a$ need not have conflict-free neighbors in $S(G,a)$ if $\mathcal{C}_{s_a}$ is such that all their conflict-free neighbors are outside of $S(G,a)$.
		Moreover, for each atom $a$ and each feasible neighborhood configuration $\mathcal{C}_{s_a}$, the algorithm computes and stores the minimum number of colored vertices required for a conflict-free $k$-coloring of $S(G,a)$ extending $\mathcal{C}_{s_a}$.
		
		If the list of feasible neighborhood configurations of any vertex or edge becomes empty at any point, the algorithm aborts and reports that the graph is not conflict-free $k$-colorable.
		Otherwise, after processing the root $r$, the algorithm checks all feasible neighborhood configurations of $s_r$ to find a neighborhood configuration for which the number of colored vertices is minimal.
		Starting with this neighborhood configuration, the algorithm backtracks and reconstructs a conflict-free $k$-coloring of $G$ with a minimal number of colored vertices. %TODO do we need more details?

		It remains to describe how to process an atom of $T(G)$.
		In case of a face atom $f$, the incoming separator can either be a vertex or an edge separator.
		We assume that it is an edge separator $e_1 = uv$; vertex separators can be handled analogously.
		The face $f$ may contain vertices and edges that are not part of any separator.
		For those vertices and edges, we have already computed the set of feasible neighborhood configurations in the first step of the algorithm.
		All other vertices and edges except for the incoming separator correspond to children of $f$ in $T(G)$; therefore, we have already computed the set of feasible neighborhood configurations for each of them.

		For each neighborhood configuration $\mathcal{C}_{e_1}$ still in the list of feasible neighborhood configurations of $e_1$, we build the directed \emph{neighborhood configuration graph} $G^* = (V^*,E^*)$ as depicted in Figure~\ref{fig:dpfaceConfigurationGraph}.
		
		The vertex set of $G^*$ consists of the neighborhood configuration $\mathcal{C}_{e_1}$ and each feasible neighborhood configuration $\mathcal{C}_{e_i}$ of each edge $e_i \neq e_1$ of $f$.
		There is an edge between two neighborhood configurations $\mathcal{C}_{uv}, \mathcal{C}_{vw}$ iff they are compatible.
		We choose a direction of the edges around the face and direct all edges in $G^*$ accordingly; see Figure~\ref{fig:dpfaceConfigurationGraph}.
		Each simple directed cycle $D_f$ in $G^*$ must contain $\mathcal{C}_{e_1}$ and thus corresponds to a selection of one neighborhood configuration for each edge of $f$; these neighborhood configurations are pairwise compatible.
		Therefore, there is a conflict-free $k$-coloring of $S(G,f)$ extending $\mathcal{C}_{e_1}$ iff there is a simple directed cycle in $G^*$.
		
		Moreover, we add weights to the edges of $G^*$ such that the weight of a simple directed cycle corresponds to the minimum number of colored vertices in such a coloring.
		In order to compute the weights, for vertices and edges of $f$ that are separators corresponding to children of $f$ in $T(G)$, we make use of the minimum number of colored vertices in their corresponding subtrees that we computed earlier.

		We can find a minimum-weight cycle in $G^*$ or decide there is no such cycle in time $\mathcal{O}(|V^*| + |E^*|)$, using an algorithm similar to Dijkstra's shortest path algorithm.
		We can do this in linear time because we can expand the vertices in fixed order, expanding all vertices corresponding to an edge of $f$ before moving on to all vertices of the next edge around the face.
		If our algorithm finds a minimum-weight cycle, we store its weight as the minimum number of vertices colored in any conflict-free $k$-coloring of $S(G,f)$ extending $\mathcal{C}_{e_1}$.
		Otherwise, $\mathcal{C}_{e_1}$ is removed from the list of feasible neighborhood configurations of $e_1$.
		Repeating this procedure for each feasible neighborhood configuration of $e_1$ concludes the processing of a face atom $f$.

		In the following, we describe how to handle an edge atom $e \in V(T(G))$.
		In this case, the incoming separator $s_e$ is a vertex separator $v$.
		For each neighborhood configuration in the list of feasible neighborhood configurations of $v$, there is at least one neighborhood configuration in the list of feasible neighborhood configurations of $e$; otherwise, we would have already deleted the neighborhood configuration.
		To compute the minimum number of colored vertices in $S(G,e)$ for some neighborhood configuration $\mathcal{C}_v$, we check for each neighborhood configuration of $e$ containing $\mathcal{C}_v$ the minimum number of colored vertices, taking into account the color of $u$ and the minimum number of vertices computed for the children of $e$ in $T(G)$.
		Repeating this for each feasible neighborhood configuration of $v$ concludes the processing of an edge atom $e$.
	}
	
	%\textcolor{red}{Otherwise, we proceed as follows: We consider the case that the last atom we processed was a face. The approach for the case that the last atom was an edge is analogous. We choose for each vertex and each edge the neighborhood configuration corresponding to the cycle $D_f$ that we computed. Now, we traverse $T$ starting from the root which is the last atom that we processed last. We consider the case that a traversed vertex $t \in T$ is a face $f$. The approach for the case that $t$ is an edge is analogous. In a previous step we already chosed a specific neighborhood configuration $\mathcal{C}$ for the incoming separator of $f$. As $\mathcal{C}$ is an feasible neighborhood configuration, there is a cycle in the neighborhood configuration graph of $f$ assigning each vertex and each edge a specific neighborhood configuration. Traversing the entire tree as described above results in a set $\mathcal{C}=\{ \mathcal{C}_1,\dots,\mathcal{C}_{|E|} \}$ of neighborhood configurations over the edges of $G$ using $k$ colors such that for every pair $e = uv$, $e^{\prime} = vw$ of incident edges, $\mathcal{C}_{u}\leftrightarrow\mathcal{C}_{v}$ and $\mathcal{C}_{v}\leftrightarrow\mathcal{C}_{w}$ holds and $\mathcal{C}_{e^{\prime}}$ is compatible with $\mathcal{C}_{e}$}.
	
%	\revised{By simply associating each neighborhood configuration with the minimum number of already colored vertices we obtain an algorithm computing a conflict-free $k$-coloring of $G$ that minimizes the number of colored vertices.}
	
\subsubsection{Correctness of the Algorithm}
	\revised{Next, we argue that our algorithm is correct, i.e., it finds a conflict-free $k$-coloring with minimum number of colored vertices iff one exists.
	In this section, we call a neighborhood configuration \emph{valid} if it can be extended into a conflict-free $k$-coloring of $G$.

	There are only two reasons for deleting a neighborhood configuration $\mathcal{C}$ from a list of feasible neighborhood configurations.
	In the first case, the deletion of $\mathcal{C}$ is a consequence of a deletion of another neighborhood configuration $\mathcal{C}'$.
	In this case, $\mathcal{C}$ is deleted because deleting $\mathcal{C}'$ has led to an incident vertex or edge without a feasible neighborhood configuration compatible to $\mathcal{C}$.
	This can never cause a valid neighborhood configuration to be deleted unless we deleted a valid neighborhood configuration $\mathcal{C}'$ first.

	In the second case, the deleted neighborhood configuration $\mathcal{C}$ belongs to an incoming separator $s_f$ of a face atom $f$ for which the algorithm finds that there is no conflict-free $k$-coloring of $S(G,a)$ extending it.
	By induction on $T(G)$, we assume that when we start processing $f$, no valid neighborhood configurations have been deleted from the list of feasible neighborhood configurations for any vertex or edge of $f$.
	Assume there was a valid neighborhood configuration $\mathcal{C}$ of $s_f$ deleted by our algorithm.
	Because $\mathcal{C}$ is valid, there is a conflict-free $k$-coloring of $S(G,f)$ extending $\mathcal{C}$.
	This yields a set of compatible neighborhood configurations for the edges of $f$ and thus a cycle in the corresponding neighborhood configuration graph $G^*$.
	This is a contradiction, because the algorithm only deletes $\mathcal{C}$ if there is no such cycle.
	We conclude that no valid neighborhood configuration is ever deleted from the list of feasible neighborhood configurations of any vertex or edge.
	Therefore, the algorithm will always find the graph to be conflict-free $k$-colorable if it is.
	In a similar manner, we can argue that the number of colored vertices used by the coloring produced by our algorithm is minimal.

	In the remainder of the section, we prove that our algorithm never produces an invalid conflict-free $k$-coloring of $G$.
	Again, the proof is by induction on $T(G)$.
	We discuss an inductive step for the case that the current atom is a face $f$ with an incoming edge separator $e_1$; the induction base and the remaining cases are analogous.
	We assume by induction that for each neighborhood configuration $\mathcal{C}_{s_a}$ of the incoming separator $s_a$ of each child $a$ of $f$, there is a conflict-free $k$-coloring of $S(G,a)$ extending $\mathcal{C}_{s_a}$.
	Let $\mathcal{C}$ be a neighborhood configuration of $e_1$ that remains feasible after processing of $f$.
	This is because there is a cycle in the corresponding neighborhood configuration graph $G^*$.
	This cycle corresponds to a set of pairwise compatible edge neighborhood configurations.
	We can construct a conflict-free $k$-coloring of $S(G,f)$ by combining the colorings induced by these neighborhood configurations and the corresponding colorings of the graphs $S(G,a)$ for children $a$ of $f$.
	At the root $r$ of $T(G)$, this yields a conflict-free $k$-coloring of $G$, because all neighbors of the incoming separator of $r$ are part of $S(G,r) = G$.
	Therefore, our algorithm never produces an invalid conflict-free coloring.
	}

\subsubsection{Runtime of the Algorithm}
	\revised{Finally, we need to analyze the running time of our dynamic programming approach.
	We begin by observing that $T(G)$ has $\mathcal{O}(n)$ atoms.
	Moreover, we observe that the number of vertex neighborhood configurations $\mathcal{C}_{v} = [\chi(v), S_{v}, \rho_v]$ of a vertex $v$ is in $\mathcal{O}(n^{k})$, as there are at most $\dbinom{|N[v]|}{k} \cdot k!$ possibilities for $S_{v}$ and $\rho_v$.
	Therefore, the number of edge neighborhood configurations $\mathcal{C}_{e} = [\mathcal{C}_{u}, \mathcal{C}_{v}]$ of an edge $e \in E$ is in $\mathcal{O}(n^{2k})$.

	Let $f = e_1e_2\ldots e_m$ be a face atom; the running time for processing face atoms dominates the running time for all other computation steps of the algorithm.
	For $f$, we build the neighborhood configuration graph $G^*$ that has $\mathcal{O}(n^{2k+1})$ vertices, because $f$ has at most $n$ edges, each with $\mathcal{O}(n^{2k})$ neighborhood configurations.
	The number of edges between the neighborhood configurations of two incident edges $e_i = uv,e_{i+1} = vw$ along $f$ is at most $\mathcal{O}(n^{3k})$ because there are only $\mathcal{O}(n^k)$ neighborhood configurations for each of the vertices $u,v$ and $w$.
	Therefore, the number of edges in $G^*$ is $\mathcal{O}(n^{3k+1})$.
	This leads to a running time of $\mathcal{O}(n^{5k+2})$, because we run a graph scan on $G^*$ for each of the $\mathcal{O}(n^{2k})$ neighborhood configurations of the incoming separator and each of the $\mathcal{O}(n)$ face atoms.}
	
	\revised{Streamlining this approach leads to a runtime of $\mathcal{O}(n^{4k+1})$.
	In particular, we modify our subroutine processing a face atom $f$ that has an incoming edge separator $e=uv$ as follows.
	For each neighborhood configuration $\mathcal{C}_v$ of $v$ we extend the neighborhood configuration graph $G^{*}$ of $f$ by considering all feasible neighborhood configurations $\mathcal{C}_{e_1}$ of $e_1$ such that $\mathcal{C}_{e_1}^v = \mathcal{C}_{v}$ holds and compute minimum-weight cycles in $G^{*}$.
	For each neighborhood configuration $\mathcal{C}_{e_1}$ of $e_1$ that is reached during an application of the the shortest path algorithm, we obtain the minimum number of vertices colored in any conflict-free coloring of $S(G,f)$ extending $\mathcal{C}_{e_1}$.
	As the number of all edge and vertex neighborhood configurations of $G$ is $\mathcal{O}(n^{3k+1})$, we obtain an overall runtime of $\mathcal{O}(n^{4k+1})$.}
	
	This concludes the proof of Theorem~\ref{thm:DPoptimalconflictfreecoloring}.

\subsection{Approximability for Three or More Colors}
\label{three_colors}
In \S~\ref{sec:conflict-free-coloring-planar-sufficiency} we stated that every planar graph is conflict-free 3-colorable. In this section we deal with conflict-free $3$-colorings of planar graphs that, additionally, minimize the number of colored vertices.

\begin{theorem}\label{thm:three_colors}
Let $k\geq 3$ and let $G$ be a planar graph. The following holds:
\begin{enumerate}
%\item {\sc{$k$-Conflict-Free Dominating Set}} is NP-complete for planar graphs for $k\geq 3$.
\item[(1)] Unless $\mbox{P} = \mbox{NP}$, there is no polynomial-time approximation algorithm providing a constant-factor approximation of $\gamma_{CF}^3(G)$ for planar graphs. {\sc{$3$-Conflict-Free~Dominating Set}} is NP-complete for planar graphs.
\item[(2)] For $k\geq 4$, {\sc{$k$-Conflict-Free Dominating Set}} is NP-complete. Also, $\gamma_{CF}^k(G) = \gamma(G)$, and the problem is fixed-parameter tractable with parameter $\gamma_{CF}^k(G)$. Furthermore, there is a PTAS for $\gamma_{CF}^k(G)$.
\item[(3)] If $G$ is outerplanar, then $\gamma_{CF}^k(G) = \gamma(G)$ and there is a linear-time algorithm to compute~$\gamma_{CF}^k(G)$.
\end{enumerate}
\end{theorem}
The proof of Theorem~\ref{thm:three_colors} is based on the following polynomial-time algorithm, which transforms a dominating set $D$ of a planar graph $G$ into a conflict-free $k$-coloring of $G$, coloring only the vertices of $D$:
%The following polynomial-time algorithm transforms a dominating set $D$ of a planar graph $G$ into a conflict-free $k$-coloring of $G$, coloring only the vertices of $D$.
%This algorithm is used in the proof of Theorem~\ref{thm:three_colors}.
Let $D$ be a dominating set of a planar graph $G$.
Every vertex \mbox{$v\in V(G) \setminus D$} is adjacent to at least one vertex in $D$.
Pick any such vertex $u \in D$ and contract the edge $uv \in E(G)$ towards $u$.
Repeat this until only the vertices from $D$ remain.
Because $G$ is planar, the graph $G^{\prime} = (D, E^{\prime})$ obtained in this way is planar, as $G^{\prime}$ is a minor of $G$.
By the $4$-coloring theorem, we can compute a proper $4$-coloring of $G^{\prime}$.

\begin{lemma}
The 4-coloring generated by this procedure induces a conflict-free 4-coloring~\mbox{of~$G$.}
\end{lemma} 
\begin{proof}
Every vertex $u \in D$ is a conflict-free neighbor to itself as its color does not appear in $N_{G}(u)$.
Let $v \in V(G) \setminus D$ be some uncolored vertex, and let $u \in D$ be the vertex that $v$ was contracted towards by the algorithm.
In $G^{\prime}$, this contraction made $u$ adjacent to all other vertices in $N_G(v) \cap D$, which guarantees that the color of $u$ is unique in $N_G(v) \cap D$.
As $V(G) \setminus D$ remains uncolored, the color of $u$ is thus unique in $N_G[v]$.
\end{proof}

\begin{proof}[Proof of Theorem~\ref{thm:three_colors}]
Proposition (1) follows from Theorem \ref{thm:inapproximability-of-gamma-cf-k} of \S~\ref{sec:general-graphs-complexity-domination-number}: The reduction used there preserves planarity and proper planar 3-coloring is NP-complete.
For (2), $\gamma^k_{CF}(G) = \gamma(G)$ implies NP-hardness in planar graphs because planar minimum dominating set is NP-hard.
Moreover, the coloring algorithm \revised{allows us to} apply any approximation scheme for planar dominating set to conflict-free $k$-coloring.
We obtain a PTAS for the conflict-free domination number by applying our coloring algorithm to the dominating set produced by the PTAS of Baker and Hill \cite{outerplanar}.
Additionally, Alber et al.~\cite{afn-ptdrds-04} proved that planar dominating set is FPT with parameter $\gamma(G)$, implying that computing the planar conflict-free domination number for $k \geq 4$ is FPT with parameter $\gamma_{CF}^k(G)$.
For (3), the class of outerplanar graphs is properly $3$-colorable in linear time and closed under taking minors.
Kikuno et al.~\cite{lineartime_outerplanar} present a linear time algorithm for finding a minimum dominating set in a series-parallel graph, which includes outerplanar graphs.
The result follows by combining this linear time algorithm with the coloring algorithm mentioned above, but using just three colors instead of four.
\end{proof}

\section{Open Neighborhoods: Planar Conflict-Free Coloring}
In this section we discuss the problem of conflict-free coloring  with open neighborhoods.
Recall that an open-neighborhood conflict-free coloring is a coloring of some vertices of a graph $G$ such that every vertex has a conflict-free neighbor in its open neighborhood $N(v)$.
In some settings, this problem is a natural alternative to the closed-neighborhood variant; for instance, when guiding a robot from one location to another, a uniquely colored beacon at the robot's current position may be insufficient.

Note that isolated vertices are problematic for this variant of conflict-free coloring; therefore, in the following, we assume that $G$ does not contain isolated vertices.
Moreover, we observe the following.
\begin{observation}
Let $G$ be a graph, $v,w \in V(G)$, and $\deg(v) = 1$, $\deg(w) = 2$.
Then, for any number $k$ of colors, in any conflict-free $k$-coloring, the unique neighbor of $v$ must be colored.
Moreover, the two neighbors of $w$ cannot have the same color.
\label{obs:open-neighborhoods-proper-coloring}
\end{observation}

This leads to a straightforward reduction from proper coloring to conflict-free coloring.
Given a graph $G$, adding an otherwise isolated neighbor to each original vertex and placing a vertex with degree 2 on every original edge yields a graph $G'$ with $\chi_O(G') = \chi_P(G)$.
See Figure~\ref{fig:example-planar-4colors} for an example of this reduction.
The resulting graph $G'$ is bipartite.
Furthermore, the reduction preserves planarity, implying that bipartite planar graphs may require at least 4 colors in a conflict-free coloring.
Moreover, even though this reduction does not necessarily preserve outerplanarity, applying it to a $K_3$ yields an outerplanar graph that requires at least 3 colors.
For \emph{bipartite} planar and outerplanar graphs, these bounds are tight.

\begin{figure}
\begin{center}
	\resizebox{0.3\linewidth}{!}{\includegraphics{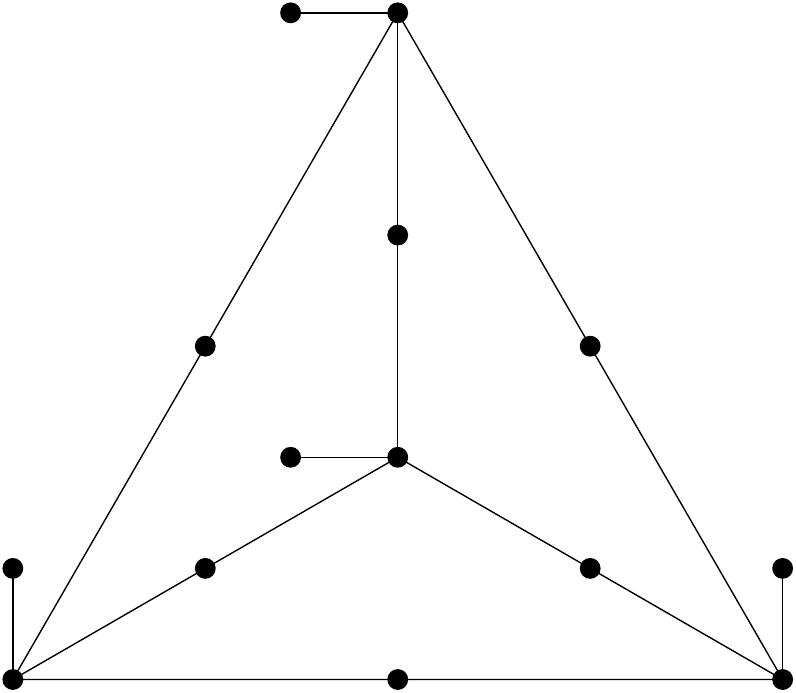}}
\end{center}
\caption{The graph $G'$ resulting from applying the reduction to $K_4$. This bipartite planar graph has $\chi_O(G') = 4$.}
\label{fig:example-planar-4colors}
\end{figure}

\begin{corollary}
It is NP-complete to decide whether a bipartite planar graph $G$ is open-neighborhood conflict-free $3$-colorable.
\end{corollary}

\begin{theorem}
	Every bipartite planar graph is open-neighborhood conflict-free $4$-colorable.
	For bipartite outerplanar graphs, three colors are sufficient.
	\label{thm:bipartite-planar-4col}
\end{theorem}
\begin{proof}
	Let $G = (V_1 \cup V_2,E)$ be a bipartite planar graph with partitions $V_1$ and $V_2$; the proof proceeds analogously for outerplanar graphs.
	We construct two minors $G_1$ and $G_2$ of $G$, to each of which we apply the planar four-color theorem.
	We build $G_1$ by merging all vertices $v \in V_2$ into an arbitrarily chosen neighbor $v_1(v) \in V_1$.
	Because $G$ is bipartite and does not contain isolated vertices, it is possible to continue this process until no vertices from $V_2$ remain.
	$G_2$ is constructed analogously, merging all vertices $v \in V_1$ into an arbitrarily chosen neighbor $v_2(v) \in V_2$.
	Each of the two resulting graphs $G_i$ contains exactly the vertices from $V_i$.
	Moreover, as a minor of $G$, $G_i$ is planar and therefore has a proper coloring with four colors.
	We assign the colors from this coloring to the vertices in $V_i$.

	It remains to show that this induces an open-neighborhood conflict-free coloring of $G$.
	Let $v$ be a vertex of $G$.
	W.l.o.g., assume $v \in V_1$.
	During the construction of $G_2$, $v$ was merged into its neighbor $v_2(v) \in V_2$.
	Therefore in $G_2$, $v_2(v)$ is adjacent to all other neighbors of $v$ in $G$.
	Because all neighbors of $v$ are in $V_2$, this implies that the color of $v_2(v)$ is unique in $N_G(v)$, and $v_2(v)$ is a conflict-free neighbor of $v$.
\end{proof}
On the other hand, for non-bipartite planar graphs, we can show the following upper bound on the number of colors.
\begin{theorem}
Every planar graph has an open-neighborhood conflict-free coloring using at most eight colors.
\end{theorem}
\begin{proof}
	Let $G = (V,E)$ be a planar graph.
	Analogous to the proof of Theorem~\ref{thm:bipartite-planar-4col} we proceed by producing two minors $G_1$ and $G_2$ of $G$, to each of which we apply the planar four-color theorem.
	However, without the assumption of bipartiteness, we cannot use the same set of four colors for $G_1$ and $G_2$, leading to a conflict-free coloring with eight colors.

	We start by constructing an independent dominating set $V_1$ of $G$.
	Let $V_2 := V \setminus V_1$.
	We construct the minor $G_i$ of $G$ by contracting each vertex \revised{$v \in V_{3-i}$} into an arbitrarily chosen neighbor $v_i(v) \in V_i$.
	Then we apply the planar four-color theorem to $G_1$ and $G_2$ with colors $\{1,2,3,4\}$ and $\{5,6,7,8\}$.
	To build a conflict-free coloring of $G$, we assign to each $v \in V_i$ its color in the proper coloring of $G_i$.
	This results in a conflict-free coloring because \revised{$v_{3-i}(v)$} is a conflict-free neighbor of $v$.
\end{proof}

Similar to the situation for closed neighborhoods, open neighborhood conflict-free coloring is hard even for $k=1$ and $k=2$.
For closed neighborhoods, a conflict-free $1$-coloring corresponds to a dominating set consisting of vertices at pairwise distance at least 3.
For open neighborhoods, a conflict-free $1$-coloring corresponds to a matching whose vertices form a dominating set and are at pairwise distance at least 3 (except for those adjacent in the matching).
\begin{theorem}
It is NP-complete to decide whether a bipartite planar graph $G$ is open-neighborhood conflict-free $1$-colorable.
\end{theorem}
\begin{proof}
	We prove hardness using a reduction from {\sc Positive Planar 1-in-3-SAT}.
	In a manner similar to the proof of Theorem~\ref{thm:planar-1-coloring-npc}, from a positive planar 3-CNF formula $\phi$ with clauses $C = \{c_1,\ldots,c_l\}$ and variables $X = \{x_1,\ldots,x_n\}$ and its plane formula graph $G(\phi)$, we construct in polynomial time a bipartite planar graph $G'_1(\phi)$ such that $\phi$ is 1-in-3-satisfiable iff $\chi_O(G'_1(\phi)) = 1$.
	The graph $G'_1(\phi)$ has one \emph{variable cycle} $v_i^0\cdots v_i^{15}$ of length $16$ for each variable $x_i$.
	There are exactly four ways to color a variable cycle; see Figure~\ref{fig:reduction-open-neighborhoods-variable}.
	Two of these color $v^0_i$ and $v^8_i$; using one of these colorings for the variable cycle of $x_i$ correspond to setting $x_i$ to {\tt true}. 
	Leaving $v^0_i$ and $v^8_i$ uncolored corresponds to setting $x_i$ to {\tt false}.
	For each clause $c_j$, $G'_1(\phi)$ contains a copy of the clause gadget depicted in Figure~\ref{fig:reduction-open-neighborhoods-variable}.
	\begin{figure}
		\begin{center}
			\begin{subfigure}[b]{.37\linewidth}
				\resizebox{\linewidth}{!}{\includegraphics{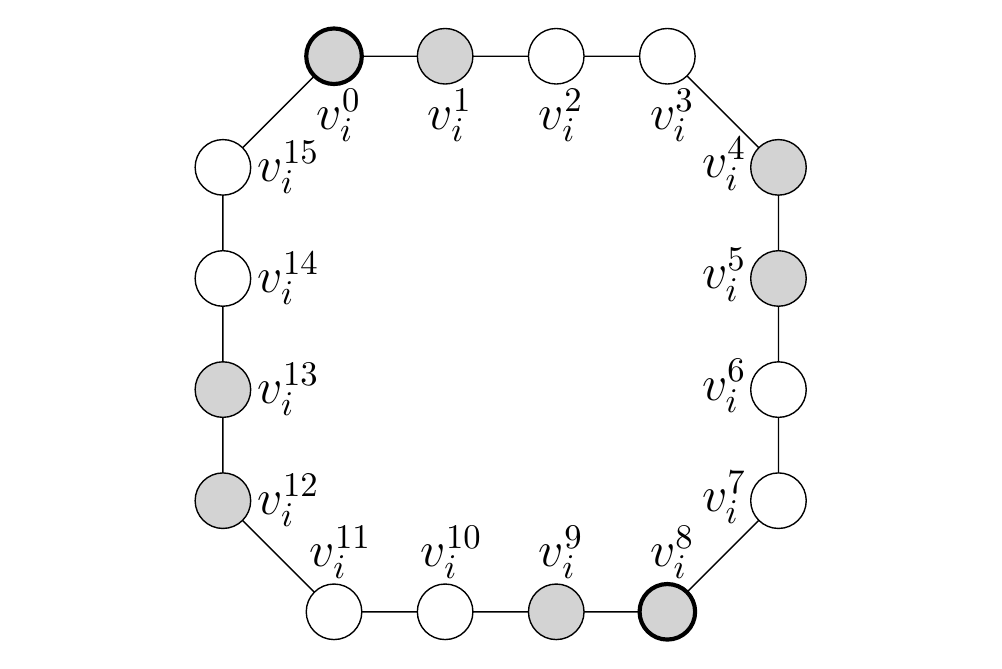}}
				\caption{A \emph{variable cycle}, with a conflict-free $1$-coloring that corresponds to setting the variable to {\tt true}.
					All conflict-free $1$-colorings of a variable cycle result from this coloring by shifting the groups of colored vertices around the cycle.
					The vertices $v_i^0$ and $v_i^8$ that may be connected to the clause gadgets are drawn with a bold outline.}
			\end{subfigure}
			\hspace{.07\linewidth}
			\begin{subfigure}[b]{.37\linewidth}
				\resizebox{\linewidth}{!}{\includegraphics{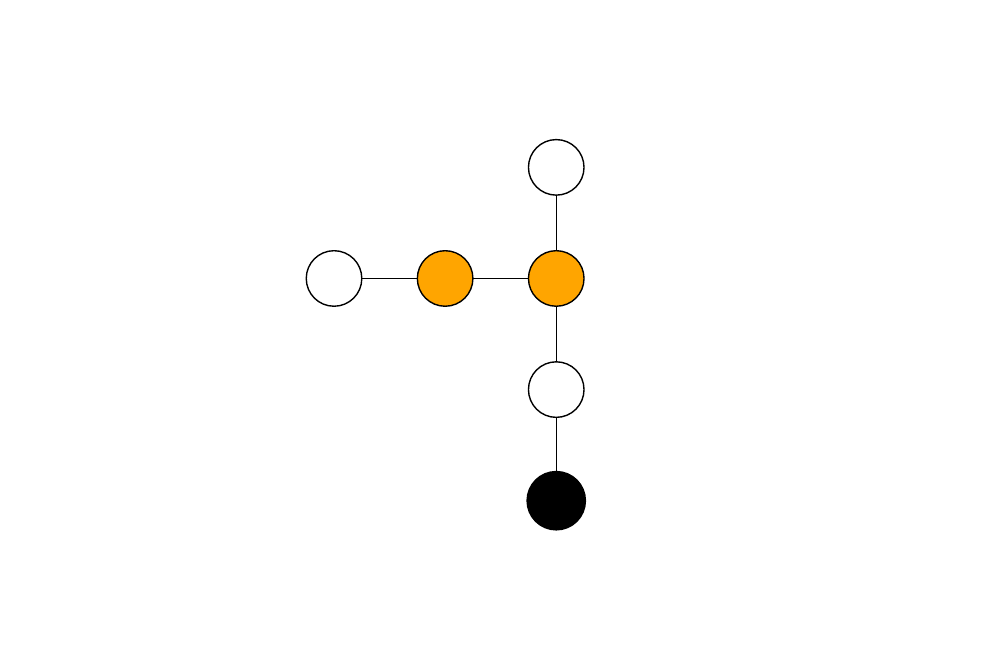}}
				\caption{A clause gadget. The orange vertices must be colored in any conflict-free $1$-coloring.
					The white vertices cannot be colored.
					The black vertex cannot be colored, but does not have a conflict-free neighbor within the gadget.
					It is connected to the variables occurring in the clause, thus enforcing that exactly one of them is set to true.
				}
			\end{subfigure}
		\end{center}
		\caption{Variable and clause gadgets for the reduction.}
		\label{fig:reduction-open-neighborhoods-variable}
	\end{figure}
	We can compute an embedding of the formula graph $G(\phi)$ in which the variable vertices are placed on a horizontal line.
	The clause vertices are embedded above and below this horizontal line.
	If a clause $c_j$ is embedded below the variables, we connect its black vertex to vertex $v_i^8$ of all variables occurring in $c_j$; otherwise, we use $v_i^0$.
	An example of this construction is depicted in Figure~\ref{fig:reduction-open-neighborhood-example}.
	\begin{figure}
		\begin{center}
		\resizebox{.75\linewidth}{!}{\includegraphics{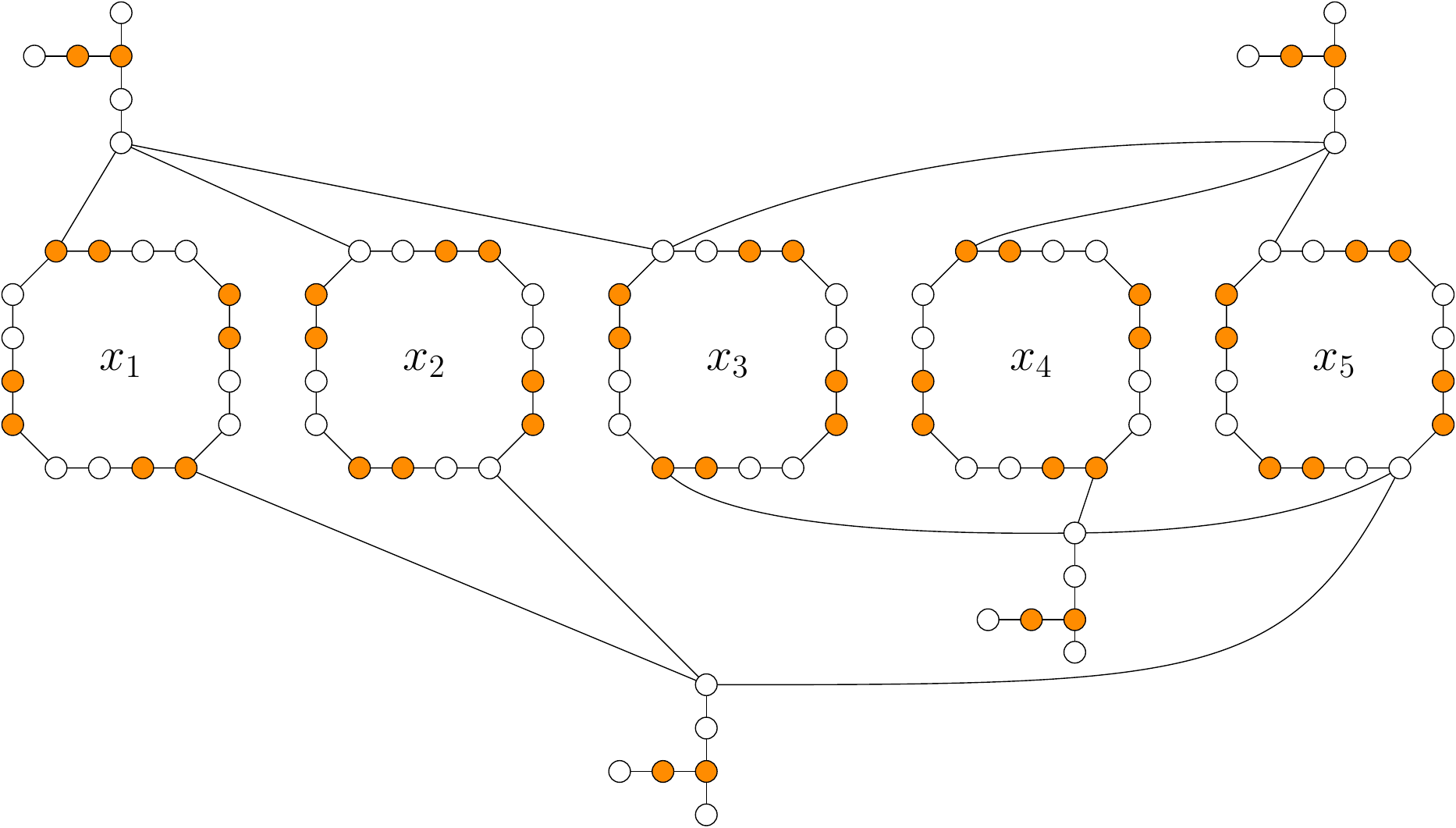}}
		\end{center}
		\caption{The graph $G'_1(\phi)$ resulting from applying the reduction to $\big\{\{x_1,x_2,x_3\},\{x_1,x_2,x_5\},\{x_2,x_4,x_5\},\{x_3,x_4,x_5\}\big\}$, and an open-neighborhood conflict-free $1$-coloring (orange vertices) corresponding to setting $x_1$ and $x_4$ to {\tt true}.}
		\label{fig:reduction-open-neighborhood-example}
	\end{figure}
	
	If $\phi$ is 1-in-3-satisfiable, coloring the variable cycles according to a satisfying assignment and the clause gadgets according to Figure~\ref{fig:reduction-open-neighborhoods-variable} yields a coloring of $G'_1(\phi)$ in which the black vertex of each clause is adjacent to exactly one colored neighbor.
	This coloring is an open-neighborhood conflict-free $1$-coloring of $\phi$.
	On the other hand, let $G'_1(\phi)$ have an open-neighborhood conflict-free $1$-coloring $\chi$.
	In each clause gadget, $\chi$ colors exactly the two orange vertices from Figure~\ref{fig:reduction-open-neighborhoods-variable}.
	Therefore, the black vertex of each clause has to be adjacent to exactly one colored variable vertex.
	Setting the variables corresponding to variable cycles with colored vertices $v^0_i$ and $v_i^8$ to {\tt true} thus yields a 1-in-3-satisfying assignment for $\phi$.
\end{proof}
The same holds for $k=2$ colors, but the restriction to bipartite planar graphs requires a slightly more sophisticated argument.
\begin{theorem}
It is NP-complete to decide whether a bipartite planar graph $G$ is open-neighborhood conflict-free $2$-colorable.
\end{theorem}
\begin{proof}
	Again we prove hardness using a reduction from {\sc Positive Planar 1-in-3-SAT}.
	From a positive planar 3-CNF formula $\phi$ with clauses $C = \{c_1,\ldots,c_l\}$ and variables $X = \{x_1,\ldots,x_n\}$ and its plane formula graph $G(\phi)$, we construct in polynomial time a bipartite planar graph $G'_2(\phi)$ such that $\phi$ is 1-in-3-satisfiable iff $\chi_O(G'_2(\phi)) \leq 2$.
	The graph $G'_2(\phi)$ has a \emph{variable path} $v_i^1v_i^2v_i^3$ of length 3 for each variable $x_i$.
	For each clause $c_j$, there is a clause gadget as depicted in Figure~\ref{fig:open-neighborhood-2col-clause-gadget}; this gadget contains a distinguished \emph{clause vertex}.
	The gadget prevents the clause vertex from being colored and cannot be used to provide a conflict-free neighbor to the clause vertex.
	We connect vertex $v_i^1$ to the clause vertex of $c_j$ with an edge iff $x_i$ occurs in $c_j$; the other vertices of clause gadgets and variable gadgets are not connected to any vertex outside their respective gadget.
	Therefore, variable vertex $v_i^1$ can provide a conflict-free neighbor to the clause vertex of $c_j$ iff $x_i$ occurs in $c_j$.
	
	We still have to enforce that the color of the conflict-free neighbor of the clause vertex is the same for all clauses.
	To this end, we connect the clause vertices using the \emph{equality gadget} depicted in Figure~\ref{fig:open-neighborhood-2col-clause-conn-gadget}.
	This gadget ensures that the conflict-free neighbors of the two clause vertices connected by it have the same color in any conflict-free $2$-coloring.
	We cannot add this gadget between all pairs of clause vertices because this would destroy planarity.
	Instead, we compute a spanning tree $T$ on the clause vertices that could be added to $G'_2(\phi)$, preserving planarity.
	Then, for each edge $c_ac_b$ of $T$, we add a copy of the equality gadget to $G'_2(\phi)$, using it to connect the clause vertices $c_a$ and $c_b$.
	Because adding the edges of $T$ preserves planarity, the graph resulting from adding the gadgets is planar as well.
	Moreover, because the equality gadget works transitively and $T$ is connected, the conflict-free neighbors of all clause vertices must receive the same color in any conflict-free $2$-coloring.
	
	It remains to prove that such a $T$ always exists.
	For this purpose, consider the plane formula graph $G(\phi)$, including the backbone of the formula.
	Because only one vertex of each variable or clause gadget is connected to vertices outside the gadget, these gadgets do not influence the planarity of $G'_2(\phi)$.
	Therefore, if adding $T$ preserves the planarity of $G(\phi)$, it also preserves the planarity of $G'_2(\phi)$.
	As root of $T$, we choose an arbitrary clause vertex $r$ on the boundary of the unbounded face of $G(\phi)$.
	We add an edge from $r$ to all other clause vertices on the boundary of the unbounded face to $T$.
	Now we consider the connected component $R$ of $r$ in $T$.
	Either $R = V(T)$, in which case we are done, or there must be a vertex $v \in R$ that lies on a face whose boundary contains a vertex $w \notin R$.
	For each such vertex $v$, we add an edge to all such vertices $w \notin R$.
	We iterate this procedure until we are done.

	\begin{figure}
		\begin{center}
			\resizebox{.4\linewidth}{!}{\includegraphics{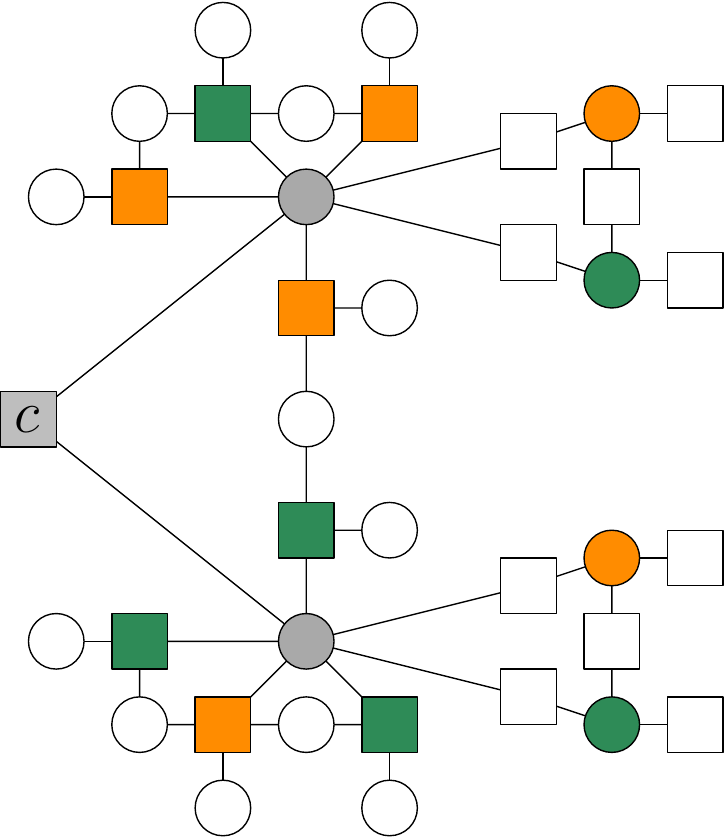}}
		\end{center}
		\caption{
			The bipartite clause gadget with clause vertex $c$; the components of the bipartition are indicated using squares and circles.
			Gray vertices cannot receive a color. Vertices colored green or orange must be colored.
			Except for automorphisms and swapping colors, orange vertices have to receive color 1 and green vertices have to receive color 2.
			White vertices may be colored or may remain uncolored; it is straightforward to extend the depicted coloring to a conflict-free $2$-coloring of the gadget (except for $c$) by coloring the white vertices of degree 1.
			By construction, one of $c$'s neighbors has three neighbors of color 1 and a conflict-free neighbor of color 2 (and vice versa for $c$'s other neighbor).
			In total, the gadget guarantees that $c$ remains uncolored and cannot have a colored neighbor within the gadget.
		}
		\label{fig:open-neighborhood-2col-clause-gadget}
	\end{figure}
	\begin{figure}
		\begin{center}
			\resizebox{.45\linewidth}{!}{\includegraphics{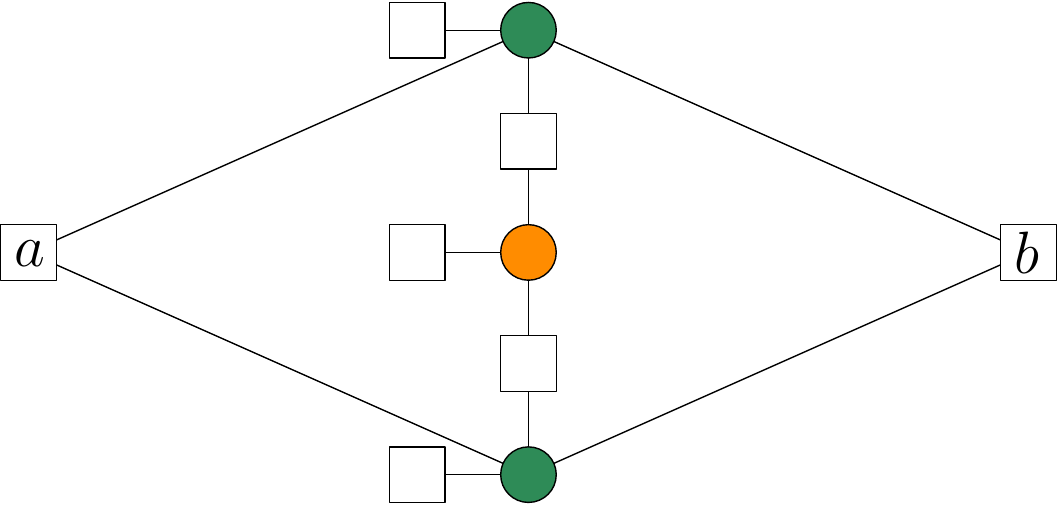}}
		\end{center}
		\caption{
			The \emph{equality gadget} that can be used to connect two terminal vertices (marked $a$ and $b$) in the same partition of a bipartite graph.
			It adds two occurrences of the same color to the neighborhoods of $a$ and $b$, thereby forcing the conflict-free neighbor of $a$ and $b$ to have the same color.
		}
		\label{fig:open-neighborhood-2col-clause-conn-gadget}
	\end{figure}

	Let $\phi$ be 1-in-3-satisfiable and let $\Gamma$ be the set of true variables in a 1-in-3-satisfying assignment of $\phi$.
	We construct a conflict-free $2$-coloring of $G'_2(\phi)$ by assigning color 1 to $v^1_i$ and $v^2_i$ for all $x_i \in \Gamma$ and to $v^3_i$ and $v^2_i$ for $x_i \notin \Gamma$.
	The vertices in equality gadgets that are adjacent to clause vertices receive color 2.
	All other vertices in the gadgets are colored as sketched in Figures~\ref{fig:open-neighborhood-2col-clause-gadget} and~\ref{fig:open-neighborhood-2col-clause-conn-gadget}.
	All clause vertices are adjacent to exactly one variable vertex carrying color 1 and thus have a conflict-free neighbor.
	Therefore, the coloring constructed in this way is a valid conflict-free $2$-coloring.

	Now assume that $G'_2(\phi)$ has a conflict-free $2$-coloring $\chi$.
	By the argument above, the conflict-free neighbor of each clause vertex is a variable vertex $v^1_i$.
	Moreover, all clause vertices have a conflict-free neighbor of the same color; w.l.o.g., color 1.
	Therefore, each clause vertex is adjacent to exactly one variable vertex with color 1, and the set of variables $x_i$ where $\chi(v^1_i) = 1$ induces a satisfying assignment of $\phi$.
\end{proof}

% !TEX root = ./main.tex
\section{Conclusion}
\label{sec:conc}

%In this paper we have given a comprehensive treatment of conflict-free coloring for
%planar graphs, but 
A spectrum of open questions remain. Many of them are related to general graphs, in particular with our sufficient condition for general graphs.
For every $k \geq 2$, $G_{k+1}$ provides an example that excluding $K_{k+2}$ as a minor is not sufficient to guarantee $k$-colorability.
However, for $k \geq 2$ we have no example where excluding $K_{k+3}^{-3}$ as a minor does not suffice.

%Variants of our problems arise from modifying the considered neighborhoods.
%In our definition of the neighborhood $N[v]$ of a vertex $v\in V$, we allow
%the vertex itself to serve as the one that is uniquely colored. In some settings
%(e.g., for guiding a robot to other locations), it is also interesting to
%require that {\em another} vertex must be uniquely colored. 
%This distinction has been dubbed ``closed'' neighborhood
%($N[v]$, including $v$) and ``open'' neighborhood ($N(v)$, excluding $v$) by
%Gargano and Rescigno~\cite{gr-ccfcg-15}. It would be interesting to expand our positive results to the 
%case of open neighborhoods; another proof of NP-completeness seems straightforward.
With respect to open-neighborhood conflict-free coloring, several open questions remain.
Are four colors always sufficient for general planar graphs?
Are three colors always sufficient for outerplanar graphs?
%Is it NP-hard to decide whether a bipartite planar graph has an open-neighborhood conflict-free $2$-coloring?

Another variant of our problem arises from requiring that {\em all} vertices must be colored.
It is clear that one extra color suffices for this purpose; however, it is not always
clear that this is also necessary, in particular, for planar graphs.
Adapting our argument to this situation does not seem straightforward, especially because there are outerplanar graphs requiring three colors in this setting.

In addition, there is a large set of questions related to geometric versions of the problem.
What is the worst-case number of colors for straight-line visibility graphs within simple polygons?
It is conceivable that $\Theta(\log\log n)$ is the right answer, just like for rectangular
visibility, but this is still an open problem, just like complexity and approximation.
Other questions arise from considering geometric intersection graphs, such as unit-disk intersection graphs,
for which necessary and sufficient conditions, just like upper and lower bounds, would be
quite interesting.

\section*{Acknowledgments}
We thank Bruno Crepaldi, Pedro de Rezende, Cid de Souza, Stephan Friedrichs, Michael Hemmer and Frank Quedenfeld for helpful discussions.
Work on this paper was partially supported by the DFG Research Unit \emph{``Controlling Concurrent Change''}, funding number FOR 1800, 
project FE407/17-2, ``Conflict Resolution and Optimization''.
\pagebreak
\bibliography{refs}

\begin{thebibliography}{10}

\bibitem{abp-ftcfc-08}
M.~A. Abam, M.~de~Berg, and S.-H. Poon.
\newblock Fault-tolerant conflict-free colorings.
\newblock In {\em Proceedings of the 20th Canadian Conference on Computational
  Geometry (CCCG)}, pages 13--16, 2008.

\bibitem{aad+-tcscf-17}
Z.~Abel, V.~Alvarez, E.~D. Demaine, S.~P. Fekete, A.~Gour, A.~Hesterberg,
  P.~Keldenich, and C.~Scheffer.
\newblock Three colors suffice: Conflict-free coloring of planar graphs.
\newblock In {\em Proceedings of the 28th Annual ACM-SIAM Symposium on Discrete
  Algorithms (SODA 2017)}, pages 1951--1963, 2017.

\bibitem{aegr-cfcrro-07}
D.~Ajwani, K.~Elbassioni, S.~Govindarajan, and S.~Ray.
\newblock Conflict-free coloring for rectangle ranges using {$O(n^{.382})$}
  colors.
\newblock In {\em Proceedings of the 19th Symposium on Parallelism in
  Algorithms and Architectures (SPAA)}, pages 181--187, 2007.

\bibitem{afn-ptdrds-04}
J.~Alber, M.~R. Fellows, and R.~Niedermeier.
\newblock Polynomial-time data reduction for dominating set.
\newblock {\em Journal of the ACM}, 51(3):363--384, 2004.

\bibitem{as-cfcsd-06}
N.~Alon and S.~Smorodinsky.
\newblock Conflict-free colorings of shallow discs.
\newblock In {\em Proceedings of the 22nd Symposium on Computational Geometry
  (SoCG)}, pages 41--43. ACM, 2006.

\bibitem{AH1}
K.~Appel and W.~Haken.
\newblock Every planar map is four colorable. {Part I. Discharging}.
\newblock {\em Illinois Journal of Mathematics}, 21:429--490, 1977.

\bibitem{AH2}
K.~Appel and W.~Haken.
\newblock Every planar map is four colorable. {Part II. Reducibility}.
\newblock {\em Illinois Journal of Mathematics}, 21:491--567, 1977.

\bibitem{adk-efmcfch-15}
P.~Ashok, A.~Dudeja, and S.~Kolay.
\newblock Exact and {FPT} algorithms for max-conflict free coloring in
  hypergraphs.
\newblock In {\em Proceedings of the 26th International Symposium on Algorithms
  and Computation (ISAAC)}, pages 271--282, 2015.

\bibitem{outerplanar}
B.~S. Baker. and M.~Hill.
\newblock Approximation algorithms for {NP}-complete problems on planar graphs.
\newblock {\em Journal of the ACM}, 41(1):153--180, 1994.

\bibitem{bco+-ocfch-10}
A.~Bar-Noy, P.~Cheilaris, S.~Olonetsky, and S.~Smorodinsky.
\newblock Online conflict-free colouring for hypergraphs.
\newblock {\em Combinatorics, Probability and Computing}, 19(04):493--516,
  2010.

\bibitem{cgrs-scfci-14}
P.~Cheilaris, L.~Gargano, A.~A. Rescigno, and S.~Smorodinsky.
\newblock Strong conflict-free coloring for intervals.
\newblock {\em Algorithmica}, 70(4):732--749, 2014.

\bibitem{css-piclcfcgh-11}
P.~Cheilaris, S.~Smorodinsky, and M.~Sulovsky.
\newblock The potential to improve the choice: list conflict-free coloring for
  geometric hypergraphs.
\newblock In {\em Proceedings of the 27th Symposium on Computational Geometry
  (SoCG)}, pages 424--432. ACM, 2011.

\bibitem{ct-gumcfc-11}
P.~Cheilaris and G.~T{\'{o}}th.
\newblock Graph unique-maximum and conflict-free colorings.
\newblock {\em Journal of Discrete Algorithms}, 9(3):241--251, 2011.

\bibitem{cfk+-ocfci-07}
K.~Chen, A.~Fiat, H.~Kaplan, M.~Levy, J.~Matousek, E.~Mossel, J.~Pach,
  M.~Sharir, S.~Smorodinsky, U.~Wagner, and E.~Welzl.
\newblock Online conflict-free coloring for intervals.
\newblock {\em SIAM Journal on Computing}, 36:1342--1359, 2007.

\bibitem{em-cfcrr-06}
K.~Elbassioni and N.~H. Mustafa.
\newblock Conflict-free colorings of rectangles ranges.
\newblock In {\em Proceedings of the 23rd Symposium on Theoretical Aspects of
  Computer Science (STACS)}, pages 254--263, 2006.

\bibitem{elrs-cfcsg-03}
G.~Even, Z.~Lotker, D.~Ron, and S.~Smorodinsky.
\newblock Conflict-free colorings of simple geometric regions with applications
  to frequency assignment in cellular networks.
\newblock {\em SIAM Journal on Computing}, 33(1):94--136, 2003.

\bibitem{gr-ccfcg-15}
L.~Gargano and A.~A. Rescigno.
\newblock Complexity of conflict-free colorings of graphs.
\newblock {\em Theoretical Computer Science}, 566:39--49, 2015.

\bibitem{gst-cfcg-14}
R.~Glebov, T.~Szab\'{o}, and G.~Tardos.
\newblock Conflict-free coloring of graphs.
\newblock {\em Combinatorics, Probability and Computing}, 23:434--448, 2014.

\bibitem{hadwiger}
H.~Hadwiger.
\newblock {\"U}ber eine {K}lassifikation der {S}trecken\-komplexe.
\newblock {\em Vierteljahresschrift der Naturforschenden Gesellschaft in
  Z\"urich}, 88:133–143, 1943.

\bibitem{hs-cfcpsrp-05}
S.~Har-Peled and S.~Smorodinsky.
\newblock Conflict-free coloring of points and simple regions in the plane.
\newblock {\em Discrete {\&} Computational Geometry}, 34(1):47--70, 2005.

\bibitem{hoffmann_et_al:LIPIcs:2015:5097}
F.~Hoffmann, K.~Kriegel, S.~Suri, K.~Verbeek, and M.~Willert.
\newblock Tight bounds for conflict-free chromatic guarding of orthogonal art
  galleries.
\newblock In {\em Proceedings of the 31st Symposium on Computational Geometry
  (SoCG)}, pages 421--435, 2015.

\bibitem{hks-cfcms-10}
E.~Horev, R.~Krakovski, and S.~Smorodinsky.
\newblock Conflict-free coloring made stronger.
\newblock In {\em Proceedings of the 12th Scandinavian Symposium and Workshop
  on Algorithm Theory (SWAT)}, pages 105--117, 2010.

\bibitem{lineartime_outerplanar}
T.~Kikuno, N.~Yoshida, and Y.~Kakuda.
\newblock A linear algorithm for the domination number of a series-parallel
  graph.
\newblock {\em Discrete Applied Mathematics}, 1983.

\bibitem{lp-cfcud-09}
N.~Lev-Tov and D.~Peleg.
\newblock Conflict-free coloring of unit disks.
\newblock {\em Discrete Applied Mathematics}, 157(7):1521--1532, 2009.

\bibitem{mr-mwtnph-08}
W.~Mulzer and G.~Rote.
\newblock Minimum-weight triangulation is {NP}-hard.
\newblock {\em Journal of the ACM}, 55(2):11, 2008.

\bibitem{pt-cfcgh-09}
J.~Pach and G.~T\'ardos.
\newblock Conflict-free colourings of graphs and hypergraphs.
\newblock {\em Combinatorics, Probability and Computing}, 18(05):819--834,
  Sept. 2009.

\bibitem{RSST}
N.~Robertson, D.~Sanders, P.~Seymour, and R.~Thomas.
\newblock The four-colour theorem.
\newblock {\em Journal of Combinatorial Theory Series B}, 70:2--44, 1997.

\bibitem{s-cpcg-03}
S.~Smorodinsky.
\newblock {\em Combinatorial Problems in Computational Geometry}.
\newblock PhD thesis, School of Computer Science, Tel-Aviv University, 2003.

\bibitem{wilson}
R.~Wilson.
\newblock {\em Four colours suffice: How the map problem was solved}.
\newblock Princeton University Press, 2013.

\end{thebibliography}
\bibliographystyle{abbrv}

\end{document}